\documentclass[11pt,a4paper,reqno,intlimits,sumlimits,fleqn]{amsart}

\usepackage{amssymb}
\usepackage{color}
\usepackage[mathscr]{euscript}
\usepackage{xspace}
\usepackage{enumerate}
\usepackage{enumitem}
\usepackage{epsfig,rotating}
\usepackage{graphicx}
\input{xy}\xyoption{all}
\usepackage{todonotes}
\usepackage{amsmath}
\usepackage{dsfont}
\usepackage{color}
\usepackage{soul}
\usepackage[longnamesfirst,round]{natbib}
\usepackage{mathtools}
\usepackage{amsmath, amstext, amsfonts, mathrsfs, amssymb}
\usepackage{bbm}	
\usepackage{longtable}
\usepackage{subfigure}
\usepackage[utf8]{inputenc}
\usepackage{booktabs}
\usepackage{fancyhdr}
\usepackage{keyval}
\usepackage{lipsum}

\usepackage[left=1.55in,right=1.55in,top=1.6in,bottom=1.3in]{geometry}

\DeclareMathAlphabet{\mathpzc}{OT1}{pzc}{m}{it}

\usepackage[bookmarksopen,pdfstartview=FitH]{hyperref}
\hypersetup{colorlinks,%
           citecolor=blue,%
           filecolor=blue,%
           linkcolor=blue,%
           urlcolor=blue}

\begin{document}

\theoremstyle{plain}
\newtheorem{theorem}{Theorem}[section]
\newtheorem{lemma}[theorem]{Lemma}
\newtheorem{proposition}[theorem]{Proposition}
\newtheorem{corollary}[theorem]{Corollary}
\newtheorem{claim}[theorem]{Claim}
\newtheorem{definition}[theorem]{Definition}

\theoremstyle{definition}
\newtheorem{remark}[theorem]{Remark}
\newtheorem{note}[theorem]{Note}
\newtheorem{example}[theorem]{Example}
\newtheorem{assumption}[theorem]{Assumption}
\newtheorem*{notation}{Notation}
\newtheorem*{assuL}{Assumption ($\mathbb{L}$)}
\newtheorem*{assuAC}{Assumption ($\mathbb{AC}$)}
\newtheorem*{assuEM}{Assumption ($\mathbb{EM}$)}
\newtheorem*{assuES}{Assumption ($\mathbb{ES}$)}
\newtheorem*{assuM}{Assumption ($\mathbb{M}$)}
\newtheorem*{assuMM}{Assumption ($\mathbb{M}'$)}
\newtheorem*{assuL1}{Assumption ($\mathbb{L}1$)}
\newtheorem*{assuL2}{Assumption ($\mathbb{L}2$)}
\newtheorem*{assuL3}{Assumption ($\mathbb{L}3$)}

\newcommand{\Law}{\ensuremath{\mathop{\mathrm{Law}}}}
\newcommand{\loc}{{\mathrm{loc}}}
\newcommand{\Log}{\ensuremath{\mathop{\mathcal{L}\mathrm{og}}}}
\newcommand{\Meixner}{\ensuremath{\mathop{\mathrm{Meixner}}}}
\newtheorem*{asem}{Assumption ($\mathbb{EM}$)}
\def\EM{\ensuremath{(\mathbb{EM})}\xspace}

\newtheorem*{vol}{Assumption ($\mathbb{VOL}$)}

\let\SETMINUS\setminus
\renewcommand{\setminus}{\backslash}

\def\stackrelboth#1#2#3{\mathrel{\mathop{#2}\limits^{#1}_{#3}}}

\def\blue{\color{blue}}
\def\red{\color{red}}

\renewcommand{\theequation}{\thesection.\arabic{equation}}
\numberwithin{equation}{section}

\newcommand\llambda{{\mathchoice
      {\lambda\mkern-4.5mu{\raisebox{.4ex}{\scriptsize$\backslash$}}}
      {\lambda\mkern-4.83mu{\raisebox{.4ex}{\scriptsize$\backslash$}}}
      
{\lambda\mkern-4.5mu{\raisebox{.2ex}{\footnotesize$\scriptscriptstyle\backslash$
}}}
      
{\lambda\mkern-5.0mu{\raisebox{.2ex}{\tiny$\scriptscriptstyle\backslash$}}}}}

\newcommand{\prozess}[1][L]{{\ensuremath{#1=(#1_t)_{0\le t\le T}}}\xspace}
\newcommand{\prazess}[1][L]{{\ensuremath{#1=(#1_t)_{0\le t\le T^*}}}\xspace}

\newcommand{\de}{{\mathrm{d}}}
\newcommand{\De}{{\mathrm{D}}}
\newcommand{\im}{{\mathrm{i}}}
\newcommand{\indik}{{1}}
\newcommand{\D}{{\mathbb{D}}}
\newcommand{\E}{{\mathbb{E}}}
\newcommand{\N}{{\mathbb{N}}}
\newcommand{\Q}{{\mathbb{Q}}}
\renewcommand{\P}{{\mathbb{P}}}
\newcommand{\dd}{\operatorname{d}\!}
\newcommand{\ii}{\operatorname{i}\kern -0.8pt}
\newcommand{\Var}{\operatorname{Var }\,}
\newcommand{\dt}{\operatorname{d}\!t}   
\newcommand{\ds}{\operatorname{d}\!s}   
\newcommand{\dy}{\operatorname{d}\!y }    
\newcommand{\du}{\operatorname{d}\!u}  
\newcommand{\dv}{\operatorname{d}\!v}   
\newcommand{\dx}{\operatorname{d}\!x}   
\newcommand{\dq}{\operatorname{d}\!q}   
\newcommand{\fpoint}{\frac{\de}{\de t} f}
\newcommand{\gpoint}{\frac{\de}{\de t} g}
\newcommand{\varthetapoint}{\frac{\de }{\de t} \vartheta}
\newcommand{\gammapoint}{\frac{\de }{\de t} \gamma} 
\newcommand{\R}{\mathbb{R}}
\newcommand{\C}{\mathbb{C}}
\newcommand{\ind}{\mathbbmss{1}}

\def\EM{\ensuremath{(\mathbb{EM})}\xspace}
\def\mg{martingale\xspace}
\def\smmg{semimartingale\xspace}
\def\smmgs{semimartingales\xspace}

\def\lmm{LIBOR market model\xspace}
\def\fpm{forward price model\xspace}
\def\mfm{Markov-functional model\xspace}
\def\alm{affine LIBOR model\xspace}
\def\lmms{LIBOR market models\xspace}
\def\fpms{forward price models\xspace}
\def\mfms{Markov-functional models\xspace}
\def\alms{affine LIBOR models\xspace}

\def\inK{\ensuremath{{k\in\mathcal{K}}}}
\def\inKn{\ensuremath{{k\in\mathcal{K}\setminus\{N\}}}}

\def\SL{\ensuremath{\mathbb{SL}}\xspace}

\def\I{\mathcal{I}}

\def\e{\mathrm{e}}
\def\lib{LIBOR\xspace}
\def\lev{L\'evy\xspace}
\def\loc{\mathrm{loc}}

\newcommand{\la}{\langle}
\newcommand{\ra}{\rangle}

\newcommand{\Norml}[1]{%
{|}\kern-.25ex{|}\kern-.25ex{|}#1{|}\kern-.25ex{|}\kern-.25ex{|}}

\newcommand{\Lip}{$\mathit{Lip}$}
\newcommand{\Int}{$\mathit{Int}$}
\newcommand{\Drift}{$\mathit{Drift}$}

\DeclarePairedDelimiter{\abs}{\lvert}{\rvert} 
\DeclarePairedDelimiter{\Abs}{\Big\lvert}{\Big\rvert} 
\DeclarePairedDelimiter{\norm}{\lVert}{\rVert}

\renewcommand{\1}{\mathds{1}}
\def\R{\ensuremath{\mathbb{R}}}


\title[Multiple Curve L\'evy Forward Price Model]{Multiple Curve L\'evy Forward Price Model allowing for negative interest rates}

\author[E. Eberlein]{Ernst Eberlein}
\author[C. Gerhart]{Christoph Gerhart}
\author[Z. Grbac]{Zorana Grbac}

\address{Department of Mathematical Stochastics, University of Freiburg, Ernst-Zermelo-Stra\ss e 1, 79104 Freiburg, Germany}
\email{eberlein@stochastik.uni-freiburg.de }

\address{Department of Mathematical Stochastics, University of Freiburg, Ernst-Zermelo-Stra\ss e 1, 79104 Freiburg, Germany}
\email{christoph.gerhart@finance.uni-freiburg.de}

\address{Laboratoire de Probabilités, Statistique et Modélisation, 
	 Universit{\'e} Paris Diderot, 75205 Paris Cedex 13, France}
\email{grbac@math.univ-paris-diderot.fr}

\keywords{}
\thanks{Financial support by a Europlace Institute of Finance grant is gratefully acknowledged.}

\date{\today}
\maketitle\frenchspacing\pagestyle{myheadings}

\begin{abstract}
In this paper we develop a framework for discretely compounding interest rates which is based on the forward price process approach. This approach has a number of advantages, in particular in the current market environment. Compared to the classical as well as the L\'evy Libor market model, it allows in a natural way for negative interest rates and has superb calibration properties even in the presence of extremely low rates. Moreover, the measure changes along the tenor structure are simplified significantly. 
These properties make it an excellent base for a post-crisis multiple curve setup. Two variants for multiple curve constructions are discussed. Time-inhomogeneous L\'evy processes are used as driving processes. An explicit formula for the valuation of caps is derived using Fourier transform techniques. Based on the valuation formula, we calibrate the two model variants to market data. 
\end{abstract}

\vspace{1cm}

Traditionally the spreads between Euribor and EONIA OIS rates were in the order of magnitude of a few basis points and therefore from the point of view of modeling could be considered to be negligible. This changed definitively with the 2007-2009 financial crisis. The beginning of the crisis can easily be dated by looking at the dynamics of these spreads for different tenors (see Figure \ref{EUR_OIS_spread}). Its graph looks like a fever chart. Depending on the specific tenors, the spreads jumped to values between 40 and 70 basis points in early August 2007. The fever chart reached its peak in mid-September 2008 with the collapse of Lehman Brothers where values beyond 200 basis points were reached. 

The market had realized that there is substantial risk where it had not been recognized before. The mechanism of choice of the Euribor panel banks is such that one could assume that these banks are essentially risk-free. With the crisis it became clear that these banks are prone to liquidity and credit risk as well and consequently this risk must be priced correctly. This is directly reflected in the spreads. The presence of these spreads forces the financial industry to revise the classical single curve fixed-income models and consider multiple curve approaches.  

In this paper we develop such a model on the basis of the forward price process (for short \textit{forward process}) in the spirit of the L\'evy forward process framework introduced in \citet{EberleinOezkan05}. Figure \ref{FRA_curves} shows the historical evolution of FRA rates starting in 2005 (subfigure (a)) through 2007, 2009, 2012, 2014 up to 2016 (subfigure (f)). These curves are obtained via bootstrapping from market quotes of deposits, forward rate agreements and swaps (for technical details see \cite{GerhartLuetkebohmert18}). Deposit rates are needed for the short end, forward rate agreements for short and mid maturities, whereas the mid and long maturity part of the term structure is derived from swap quotes. As one can clearly see, in 2007 the risky tenor-dependent curves (three months and six months) started to depart from the basic discount curve. Figures (c)-(f) show a significant spread between the basic and the three-month curve and the three-month and the six-month curve. In addition we emphasize that all three term structures are negative up to some maturity in 2016. Therefore the model to be developed has to be able to cope with negative rates and tenor-dependent term structures. The curves as shown in Figure \ref{FRA_curves} represent the starting values for the model.

As in the Libor market model (LMM) in the following the interest rates are constructed via backward induction along a discrete tenor structure. The forward process approach is chosen because it is analytically, as well as numerically, superior to the LMM, which is  the industry standard as a model for discretely compounded interest rates. From the analytical point of view, the main advantage is that a more tractable measure change technique applies, which preserves the structure of the driving process and consequently allows to avoid any approximation such as the frozen drift assumption. From the economic perspective, in view of the current market environment it has to be underlined that the forward process approach allows in a natural way for negative interest rates. Since the basic quantity which is modelled is a scaled and shifted interest rate, this approach is similar in spirit to the shifted LMM, which has become the industry response to the current situation of low and negative rates. Whereas in the shifted LMM an arbitrary choice or statistical estimation of the lower boundary for negative values is required, here the range of negative rates arises from the definition of the forward process as a positive quantity (see equation (\ref{forward_discount_d}) below).
 
Another important property is related to calibration. The increments of the driving process translate directly into increments of the interest rates, which allows for superb calibration results. This is not the case for the LMM, where the increments of the driving process are scaled by the current level of the rates (for explicit expressions see the introduction of \citet*{EberleinEddahbiLalaoui16}). In particular, in a market with extremely low rates this creates serious problems. Huge movements of the driving process are needed even for small changes of the rates (exploding volatility in calibration). As far as the use of Lévy processes is concerned we mention also that - compared to models driven by a Brownian motion - these processes are flexible enough to generate the empirically observable levels of correlations between rates with different maturities. This aspect has been intensively studied in \cite{BeinhoferEberlein2011}. Let us emphasize again that the forward process approach has in addition the advantage of a smooth measure change along the tenor structure.  

Among the more recent papers which deal with multiple curve modelling we mention \cite{CuchieroFontanaGnoatto16, CuchieroFontanaGnoatto17}. In the first one an approach based on multiplicative spreads with semimartingales as driving processes is developed whereas in the second one the authors study an affine model setup. Some constraints on the state space of the driving process are needed in the latter one in order to ensure the desired behavior of the multiple curve spreads such as positivity and monotonicity of the tenor-dependent curves.  The Lévy forward price framework instead allows for full flexibility in choosing the driving process. Another multiple curve affine LIBOR model with a focus on the positivity of spreads is developed in \cite{Grbac15affine}. Furthermore, we mention \cite{CrepeyGrbacNguyen12} where a single risky rate is considered in addition to the basic one in a HJM-framework. In \cite{EberleinGerhart17} a fully fledged model with an arbitrary number of tenor-dependent curves is developed. In addition this model considers multiple curves in the context of a two-price economy and therefore allows to exploit bid and ask quotes.

For an exhaustive literature overview on multiple curve models we mention the two monographs by \cite{Henrard14} and \cite{GrbacRunggaldier15} and the article collection by \cite{BianchettiMorini13}. The idea of multiplicative spreads was introduced in \cite{Henrard10} and further used in \cite{CuchieroFontanaGnoatto16, CuchieroFontanaGnoatto17}, with spreads defined on continuous tenor structures in the context of a short rate (spot spreads) and a Heath-Jarrow-Morton framework. 

In this paper the discretely compounded multiplicative forward spreads are modelled on a discrete tenor structure extending in a suitable manner the forward process approach.  The paper is organized as follows. In Section 1 we introduce the driving process and discuss its main properties. The model is presented in Section 2. Section 3 deals with interest rate option pricing. The last section is dedicated to the implementation and calibration of the model.

\begin{figure}
\hfill
\subfigure[]{\includegraphics[width=5cm]{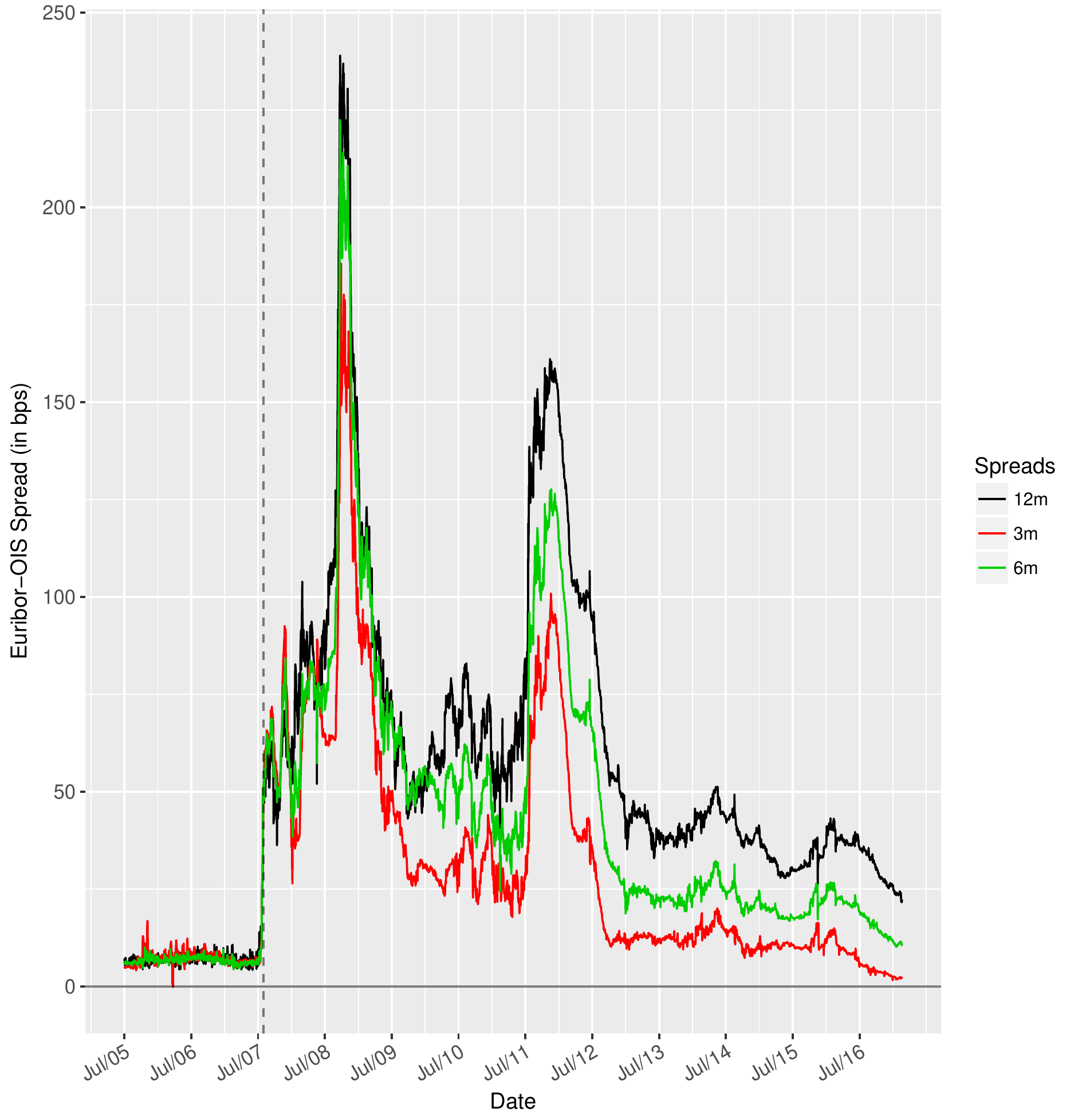}}
\hfill
\subfigure[]{\includegraphics[width=5cm]{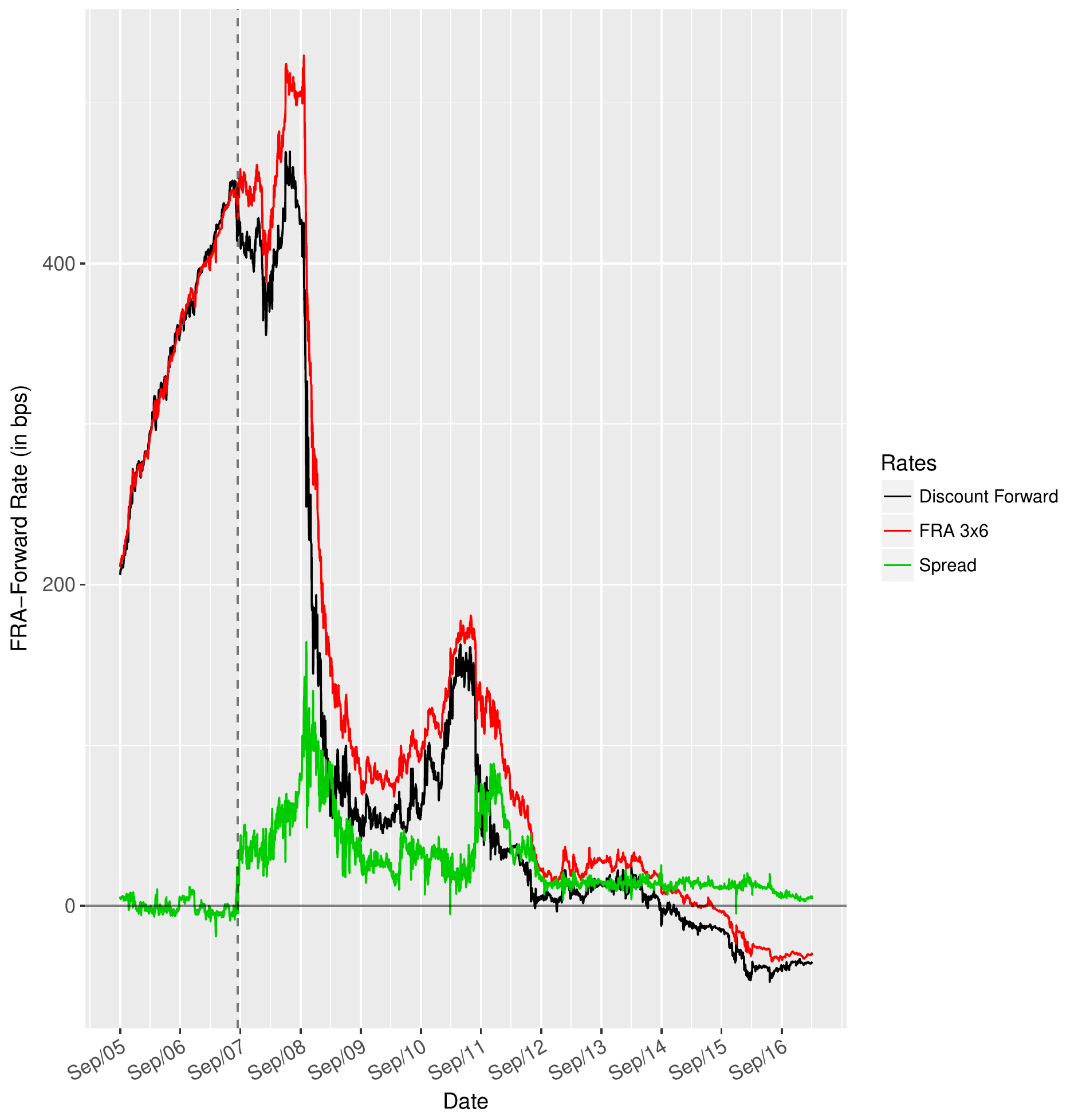}}
\hfill
\caption{Left panel: Evolution of the spreads between EURIBOR and EONIA OIS rates for different tenors. Right panel: Divergence of FRA rate and forward rate.}
\label{EUR_OIS_spread}
\end{figure}

\begin{figure}
\subfigure[]{\includegraphics[width=5.5cm]{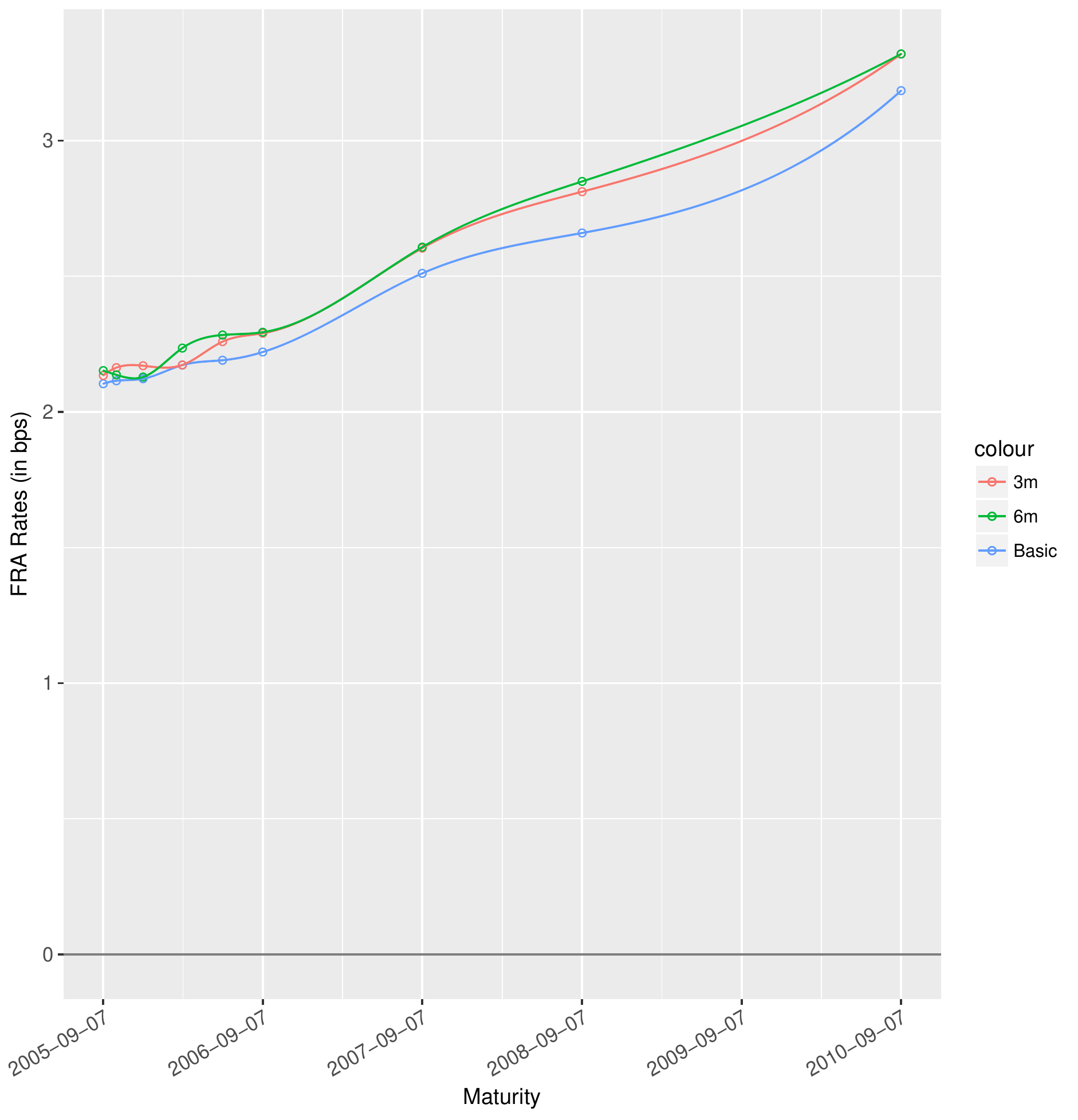}}
\hfill
\subfigure[]{\includegraphics[width=5.5cm]{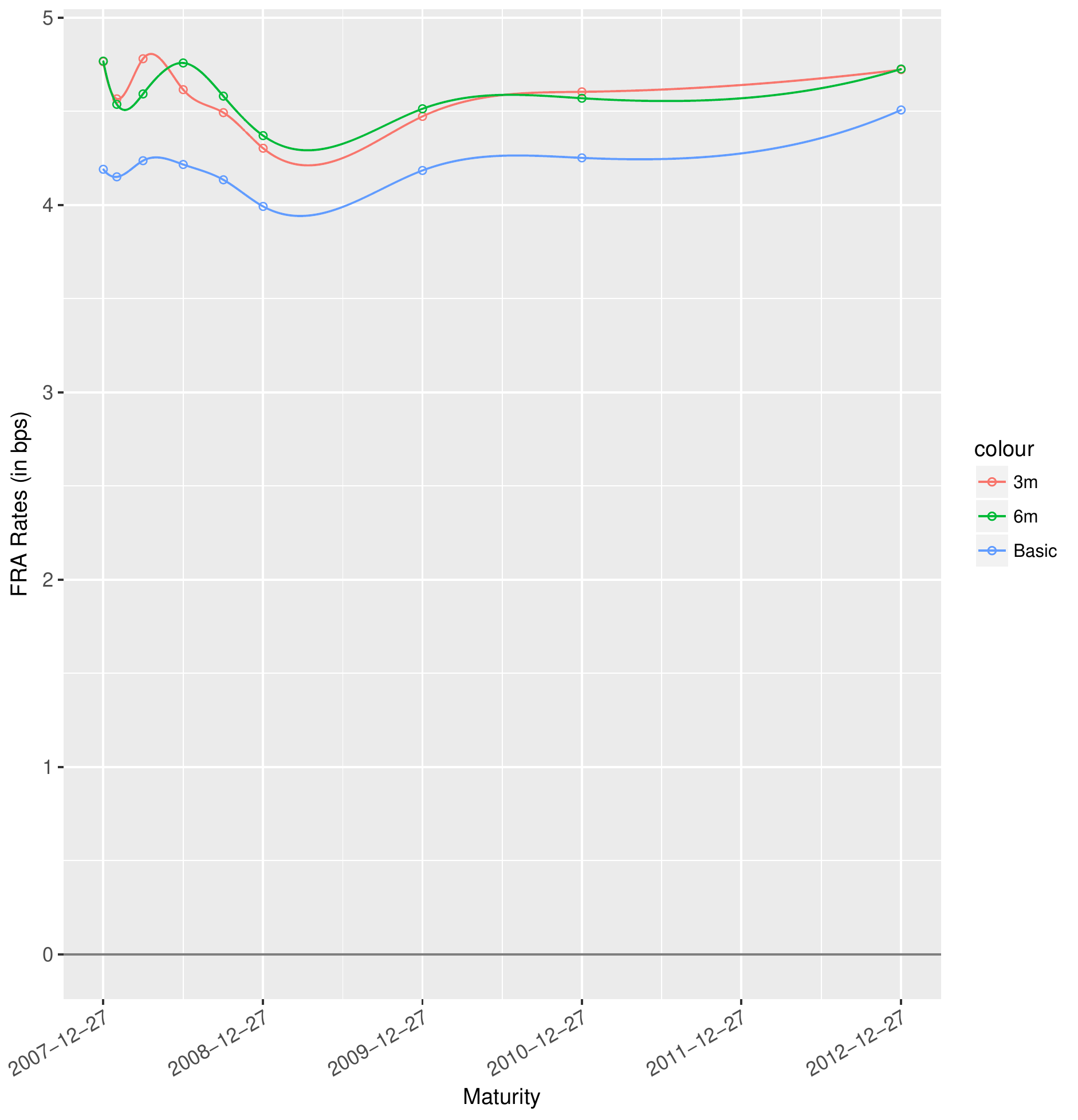}}
\hfill
\subfigure[]{\includegraphics[width=5.5cm]{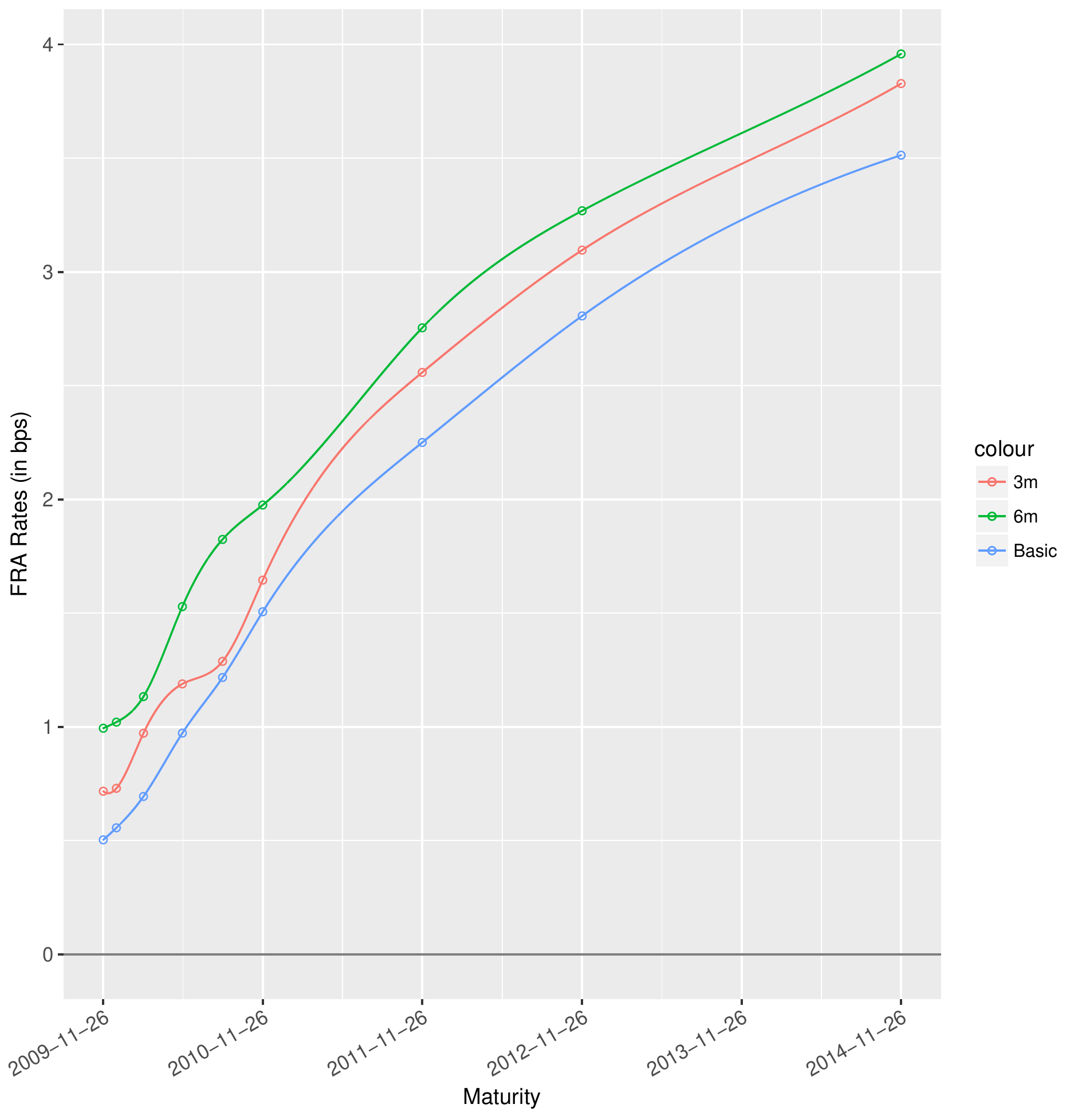}}
\hfill
\subfigure[]{\includegraphics[width=5.5cm]{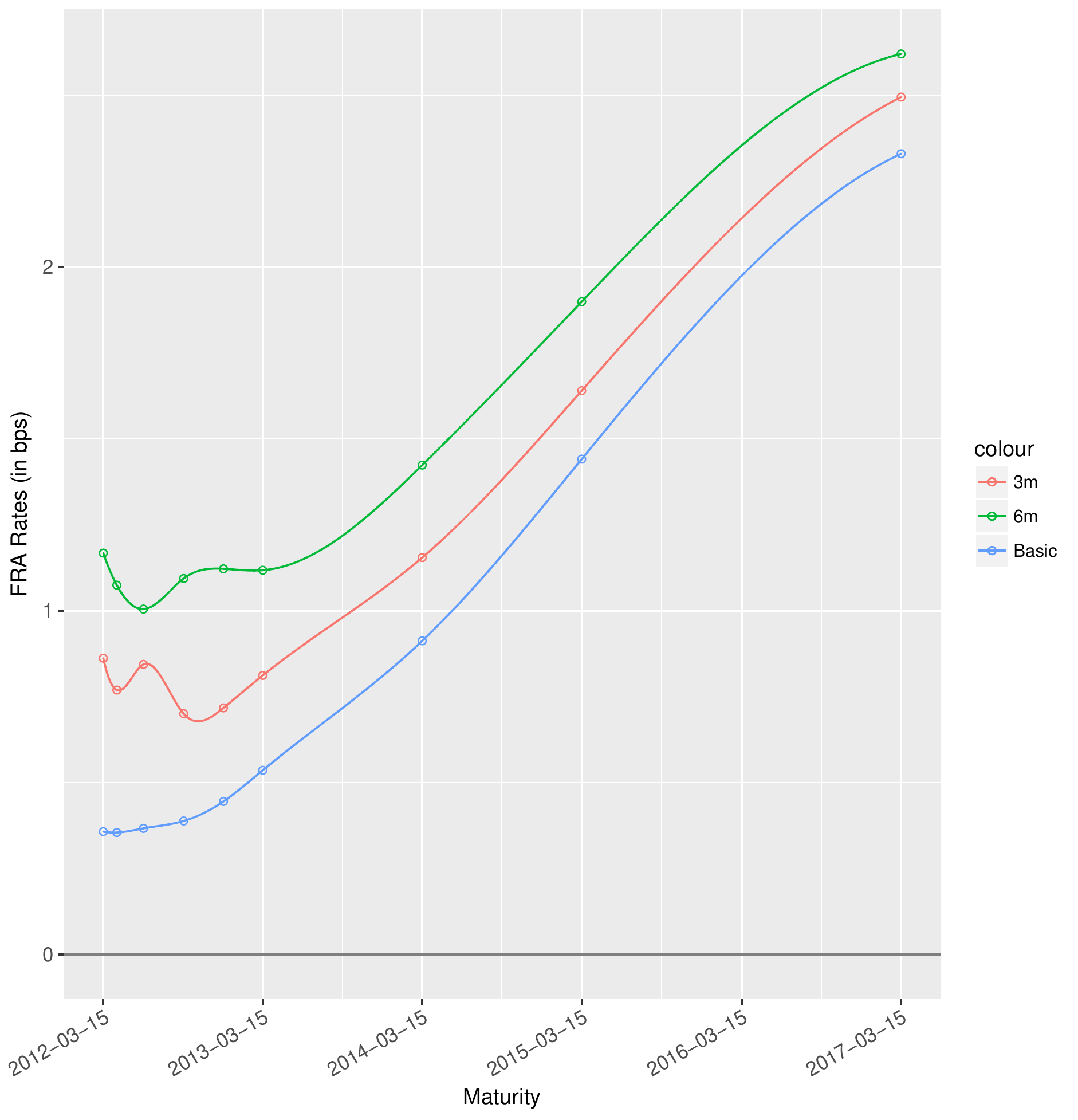}}
\hfill
\subfigure[]{\includegraphics[width=5.5cm]{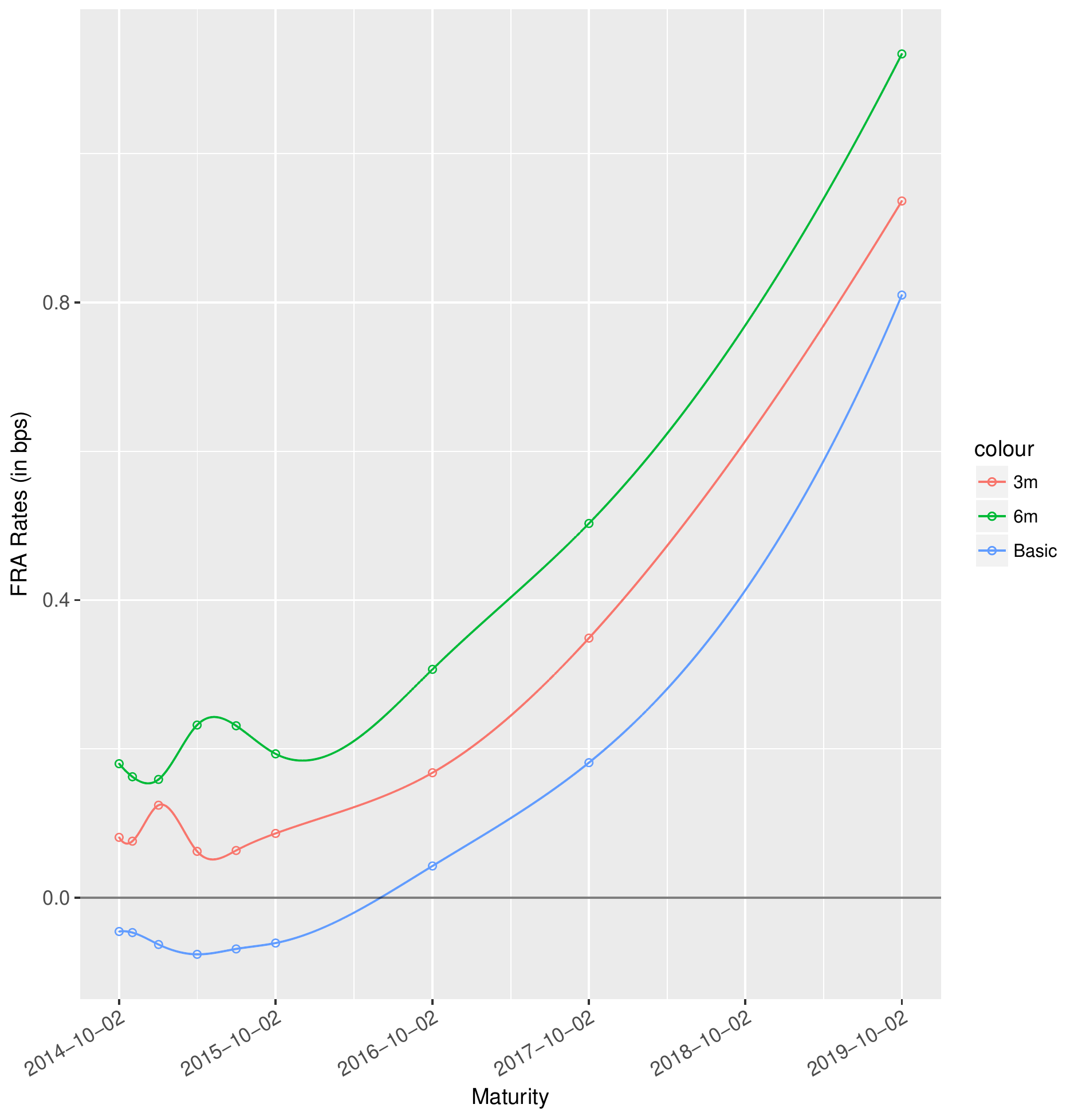}}
\hfill
\subfigure[]{\includegraphics[width=5.5cm]{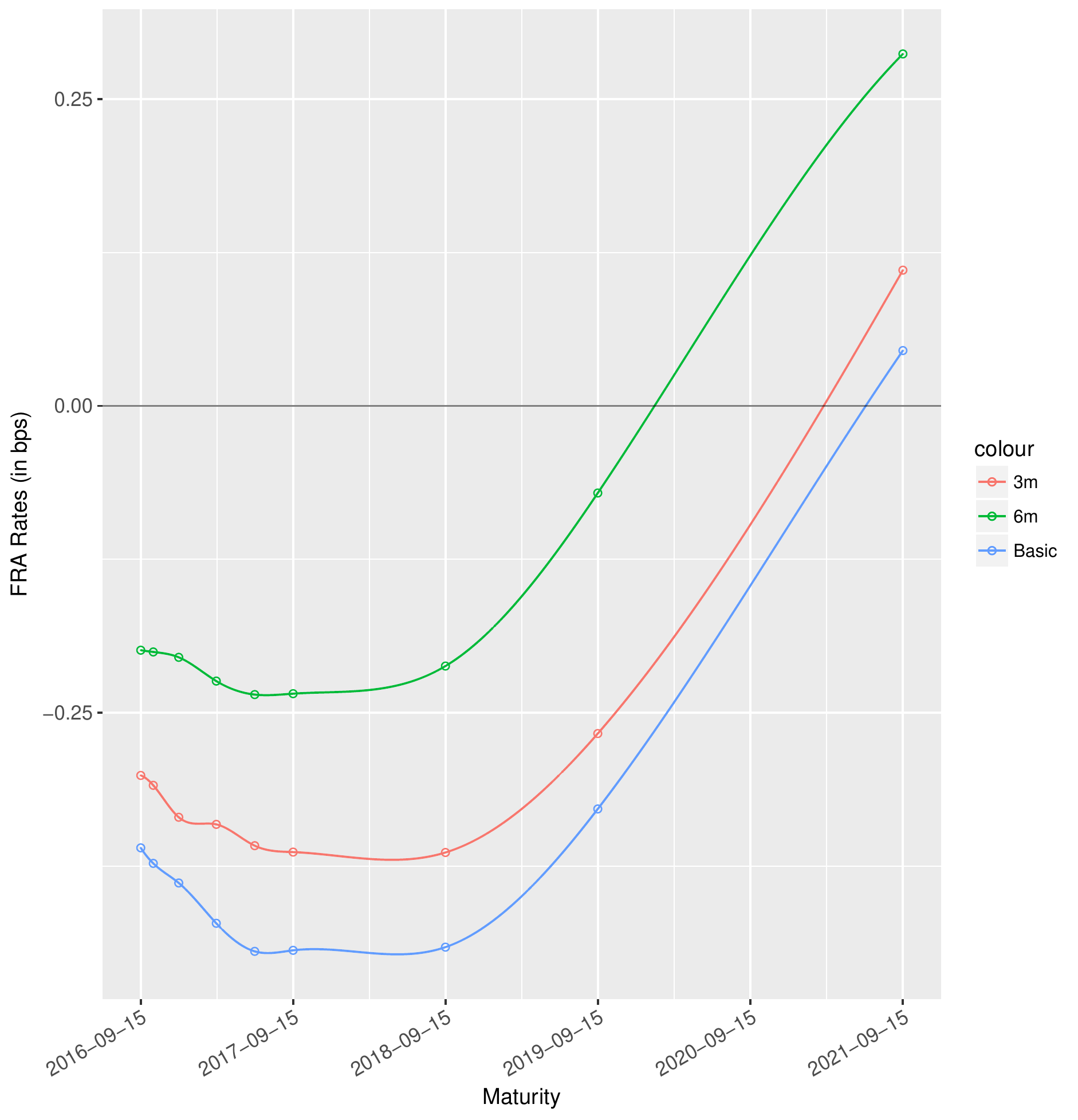}}
\hfill
 \caption{Historical evolution of the bootstrapped tenor-dependent FRA curves from market data.}
\label{FRA_curves}
\end{figure}

\section{The driving process}

Let $T^* \in \R_+ \coloneqq [0,\infty)$ be a finite time horizon and $\mathscr{B} \coloneqq (\Omega,\mathscr{G},\mathbb{F}=(\mathscr{F}_t)_{t \in [0,T^*]},P)$ a stochastic basis that satisfies the usual conditions in the sense of \citet[][Definition I.1.2 and Definition I.1.3]{JacodShiryaev03}. As driving process, we consider a $d$-dimensional time-inhomogeneous Lévy process $L=(L^1,\dots,L^d)$ on $\mathscr{B}$ with $L^i=(L_t^i)_{t \in [0,T^*]}$ for every $i \in \{1,\dots,d\}$. This means that $L$ is an $\mathbb{F}$-adapted process with independent increments and absolutely continuous characteristics (abbreviated as PIIAC, see \citet{JacodShiryaev03}). This type of stochastic processes is also known as additive processes (see \citet[][]{Sato99}). We emphasise that $L$ is a $d$-dimensional semimartingale.

Without loss of generality we can assume that the paths of each component of $L$ are càdlàg. We also postulate that each component $L^i$ starts in zero. The law of $L_t$ is determined by its characteristic function
\begin{align} \label{characteristic_function_L_t}
\E[e^{i \langle u, L_t \rangle}]= \exp & \Bigg( \int_0^t \Big[  i \langle u, b_s(h) \rangle - \frac{1}{2} \langle u, c_s u\rangle \nonumber \\
& +  \int_{\R^d} \left( e^{i \langle u, x \rangle} - 1 - i \langle u, h(x) \rangle \right) F_s(dx)  \Big] ds \Bigg) \quad (u \in \R^d).
\end{align} 
Here, $h$ is a truncation function, where usually one takes $h(x)=x \cdot \mathds{1}_{\{\abs{x} \leq 1\}}$, $b_s(h) = (b_s^1(h),\dots,b_s^d(h)): [0,T^*] \rightarrow \R^d$, $c_s=(c_s^{i j})_{ i,j \leq d}: [0,T^*] \rightarrow \R^{d \times d}$, a symmetric nonnegative-definite $d \times d$-matrix and $F_s$ is a Lévy measure for every $s \in [0,T^*]$, i.e. a nonnegative measure on $(\R^d,\mathcal{B}(\R^d))$ that integrates $(\abs{x}^2 \wedge 1)$ and satisfies $F_s(\{0\})=0$. We denote by $\langle \cdot, \cdot\rangle$ the Euclidean scalar product on $\R^d$ and $\abs{\cdot}$ is the corresponding norm. The scalar product on $\R^d$ is extended to complex numbers by setting $\langle w, z \rangle \coloneqq \sum_{j=1}^d w_j z_j$ for every $w, z \in \C^d$. Thus, $\langle \cdot, \cdot \rangle$ is not the Hermitian scalar product here. We further assume that 
\begin{align*}
\int_0^{T^*} \Big[ \abs{b_s(h)} + \norm{c_s} + \int_{\R^d}(\abs{x}^2 \wedge 1)F_s(dx) \Big] ds < \infty,
\end{align*}
where $\norm{\cdot}$ denotes any norm on the set of $d \times d$-matrices. The triplet $(b,c,F) = (b_s,c_s,F_s)_{s \in [0,T^*]}$ represents the local characteristics of $L$. We also make the following standing assumption on the exponential moments.
\begin{asem} \label{EM}
There exist constants $M, \varepsilon >0$ such that
$$
\int_{0}^{T^{*}} \int_{|x| >1} \exp \la u, x\ra F_{t} (\dx) \dt < \infty,
$$
for every $u \in [-(1+\varepsilon)M, (1+\varepsilon)M]^{d}$. In
particular, we assume without loss of generality that $\int_{|x| >1}
\exp \la u, x\ra F_{t} (\dx) < \infty $, for \emph{all} $t \in [0,
T^{*}]$.
\end{asem}

Assumption $(\mathbb{EM})$ is equivalent to $\E[\exp \langle u , L_t \rangle] < \infty$ for all $t \in [0,T^*]$ and $u \in [-(1+\epsilon)M,(1+\epsilon)M]^d$. We will consider interest rate models with underlying processes that are exponentials of stochastic integrals with respect to $L$. These  underlying processes have to be martingales under the risk-neutral measure. Therefore, a priori they have to have finite expectations which is exactly guaranteed by assumption $(\mathbb{EM})$. An immediate consequence of $(\mathbb{EM})$ is that the random variable $L_t$ has finite expectation. Therefore, the representation (\ref{characteristic_function_L_t}) simplifies and can be written as
\begin{align} \label{characteristic_function_L_tb}
\E[e^{i \langle u, L_t \rangle}]= \exp \Bigg( \int_0^t & \Big[ i \langle u, b_s \rangle - \frac{1}{2} \langle u, c_s u\rangle \nonumber \\
& +  \int_{\R^d} \left( e^{i \langle u, x \rangle} - 1 - i \langle u, x \rangle \right) F_s(dx) \Big] ds \Bigg).
\end{align} 

We emphasise that the characteristic $b$ is now different from the one in (\ref{characteristic_function_L_t}). We will always work with the local characteristics $(b,c,F)$ that appear in (\ref{characteristic_function_L_tb}). Another implication of assumption $(\mathbb{EM})$ is that the process $L$ is a special semimartingale. Thus, its canonical representation is given by the simple form
\begin{align} \label{canonical_representation_L_t}
L_t = \int_0^t b_s ds + \int_0^t \sqrt{c_s} dW_s + \int_0^t \int_{\R^d}x(\mu^L -\nu)(ds,dx)
\end{align}
(see \citet[][Corollary II.2.38]{JacodShiryaev03}), where $W=(W_t)_{t \in [0,T^*]}$ is a standard $d$-dimensional Brownian motion, $\sqrt{c_s}$ is a measurable version of the square root of $c_s$, and $\mu^L$ is the random measure of jumps of $L$ with compensator $\nu(ds,dx)=F_s(dx)ds$. Obviously, the integrals in (\ref{canonical_representation_L_t}) should be understood componentwise. We stress that assumption $(\mathbb{EM})$ is valid for all processes of interest in implementing the model. In particular $(\mathbb{EM})$ holds  for processes that are generated by generalised hyperbolic distributions. The (extended) cumulant process associated with the process $L$ under the probability measure $P$ is denoted by $\theta_s$ and given by
\begin{align*}
\theta_s(z)=\langle z, b_s \rangle + \frac{1}{2} \langle z, c_s z \rangle + \int_{\R^d} \left( e^{\langle z, x \rangle} - 1 - \langle z, x \rangle \right) F_s(dx)
\end{align*}
for every $z \in \C^d$ where this function is defined which requires that $\mathsf{Re}(z) \in \left[- (1+\epsilon)M, (1+\epsilon)M \right]^d$. A detailed analysis of the cumulant process for semimartingales is given by \citet{KallsenShiryaev02}. Note that if $L$ is a (homogeneous) Lévy process, i.e. if the increments of $L$ are stationary, the triplet $(b_s, c_s, F_s)$ and thus also $\theta_s$ do not depend on $s$. In this case, we write $\theta$ for short. It then equals the cumulant (also called log moment generating function) of $L_1$.

\section{Multiple curve Lévy forward price model}

\subsection{Basic curve} \label{basic_curve}

Let $\mathscr{T} \coloneqq \{T_0,\dots,T_n\}$ denote an arbitrary discrete tenor structure with $n \in \N$ and $0 \leq T_0 < T_1 < \dots < T_n=T^*$. For any $k \in \{1,\dots,n\}$ we set $\delta^k \coloneqq \delta(T_{k-1},T_k)$  to be the year fraction between dates $T_{k-1}$ and $T_k$ according to a specified day count convention \cite[see][section I.1.2]{BrigoMercurio06}. We assume that the tenor structure is equidistant and we may therefore specify $\delta \coloneqq \delta^k$ for every $k \in \{1,\dots,n\}$. Thus, the considered discrete tenor structure $\mathscr{T}$ is unambiguously related to tenor $\delta$.

We denote by $B_t^\mathsf{d}(T)$ the price at time $t$ of a bond maturing at $T$. For each pair of consecutive dates $T_{k-1}, T_k \in \mathscr{T}$ with $k \in \{1,\dots,n\}$, we define the discretely compounded forward reference rate at time $t \leq T_{k-1}$ by
\begin{align*}
L^\mathsf{d}(t,T_{k-1},T_k) \coloneqq \frac{1}{\delta}\left( \frac{B_t^\mathsf{d}(T_{k-1})}{B_t^\mathsf{d}(T_k)} - 1 \right)
\end{align*}
and the forward price corresponding to this reference rate is specified as
\begin{align*}
F^\mathsf{d}(t,T_{k-1},T_k) \coloneqq \frac{B_t^\mathsf{d}(T_{k-1})}{B_t^\mathsf{d}(T_k)}.
\end{align*}
Obviously, we have the relation
\begin{align} \label{forward_discount_d}
F^\mathsf{d}(t,T_{k-1},T_k) = 1 + \delta L^\mathsf{d}(t,T_{k-1},T_k).
\end{align}

The interest rate curve corresponding to the reference rate $L^\mathsf{d}$ will be referred to as basic or discount curve. Hereafter, the aim is to develop a tractable model for the forward price process $F^\mathsf{d}(\cdot,T_{k-1},T_k)$. Note that modelling the forward price processes means specifying the dynamics of ratios of successive bond prices. 

Let $L^{T^*}=(L^{1,T^*},\dots,L^{d,T^*})$ be a time-inhomogeneous Lévy process defined on the stochastic basis $(\Omega,\mathscr{F}_{T^*},\mathbb{F}=(\mathscr{F}_t)_{t \in [0,T^*]},P_{T^*}^\mathsf{d})$ with local characteristics $(0,c_t,F_t^{T^*})$ and satisfying the exponential moment condition $(\mathbb{EM})$. This process will be used as driving process of the model. We interpret the probability measure $P_{T^*}^\mathsf{d}$ as the forward martingale measure associated with the basic curve and settlement date $T^*$. The following two ingredients are needed to develop the model for the basic curve.
\begin{enumerate}
\item[($\mathbb{DFP}.1$)] The initial term structure of bond prices $B_0^\mathsf{d}$ defined by
\begin{align*}
B_0^\mathsf{d}:
\begin{cases}
[0,T^*] \rightarrow (0,\infty) \\
T \mapsto B_0^\mathsf{d}(T)
\end{cases}
\end{align*}
is given.
\end{enumerate}
The bootstrapping method considered by \citet{AmetranoBianchetti13} can be used to construct the initial term structure. One typically takes the quotes of OIS rates to derive this curve. The starting values of the forward processes are then obtained by the relation
\begin{align} \label{discount_forward_process}
F^\mathsf{d}(0,T_{k-1},T_k)=\frac{B_0^\mathsf{d}(T_{k-1})}{B_0^\mathsf{d}(T_k)}, \quad \text{for every}~ k \in \{1,\dots,n\}.
\end{align}

\begin{enumerate}
\item[($\mathbb{DFP}.2$)] For any maturity $T_{k-1} \in \mathscr{T}$ with $k \in \{1,\dots,n\}$ there is a bounded, continuous and deterministic function $\lambda^\mathsf{d}(\cdot,T_{k-1})$ given by
\begin{align*}
\lambda^\mathsf{d}(\cdot,T_{k-1}):
\begin{cases}
[0,T^*] \rightarrow \R_+^d \\
t \mapsto \lambda^\mathsf{d}(t,T_{k-1})=(\lambda^{\mathsf{d},1}(t,T_{k-1}),\dots,\lambda^{\mathsf{d},d}(t,T_{k-1}))
\end{cases}
\end{align*}
which represents the volatility of the forward process $F^\mathsf{d}(\cdot,T_{k-1},T_k)$. We require that
\begin{align*}
\sum_{k=1}^n \lambda^{\mathsf{d},j}(t,T_{k-1}) \leq M, \quad \text{for all}~ t \in [0,T^*]~\text{and}~j \in \{1,\dots,d\},
\end{align*}
where $M$ is the constant from assumption $(\mathbb{EM})$, and we set $\lambda^\mathsf{d}(t,T_{k-1})=(0,\dots,0)$ for $t > T_{k-1}$.
\end{enumerate}

To construct the forward price processes corresponding to the basic curve we proceed by backward induction as in \citet{EberleinOezkan05}. To this end we start with the most distant tenor period $[T_{n-1},T_n]$ and, for any $t \leq T_{n-1}$, we postulate that
\begin{align} \label{discounting_forward_process_oexp}
\frac{F^\mathsf{d}(t,T_{n-1},T_n)}{F^\mathsf{d}(0,T_{n-1},T_n)} = \exp\left( \int_0^t \lambda^\mathsf{d}(s,T_{n-1})dL_s^{T_n} + \int_0^t b^\mathsf{d}(s,T_{n-1},T_n)ds \right)
\end{align}
where the time-inhomogeneous Lévy process $L^{T_n}=L^{T^*}$ can be written in terms of its canonical representation
\begin{align*}
L_t^{T^*} = \int_0^t \sqrt{c_s}dW_s^{T^*} + \int_0^t \int_{\R^d} x (\mu^L - \nu^{T^*})(ds,dx)
\end{align*}
with $d$-dimensional $P_{T^*}^\mathsf{d}$-standard Brownian motion $W^{T^*} \coloneqq (W_t^{T^*})_{t \in [0,T^*]}$ and integer-valued random measure $\mu^L$ associated to the jumps of $L^{T^*}$ having $P_{T^*}^\mathsf{d}$-compensator $\nu^{T^*}(dt,dx) \coloneqq F_t^{T^*}(dx)dt$.  

Now the drift term $b^\mathsf{d}(\cdot,T_{n-1},T_n)$ is specified in such a way that the forward price process $F^\mathsf{d}(\cdot,T_{n-1},T_n)$ becomes a $P_{T^*}^\mathsf{d}$-(local) martingale. This is achieved by setting
\begin{align} \label{drift_specification}
b^\mathsf{d}(t,T_{n-1},T_n) = & - \frac{1}{2} \langle \lambda^\mathsf{d}(t,T_{n-1}), c_t \lambda^\mathsf{d}(t,T_{n-1})^\mathsf{T} \rangle \nonumber \\
& - \int_{\R^d} \left(e^{\langle \lambda^\mathsf{d}(t,T_{n-1}), x \rangle} - 1 - \langle \lambda^\mathsf{d}(t,T_{n-1}), x \rangle \right) F_t^{T^*}(dx).
\end{align}

Applying \citet[][Theorem II.8.10]{JacodShiryaev03}, one can express the forward price process presented by an ordinary exponential (\ref{discounting_forward_process_oexp}) as a stochastic exponential
\begin{align} \label{discounting_forward_process_sexp}
F^\mathsf{d}(t,T_{n-1},T_n) = F^\mathsf{d}(0,T_{n-1},T_n) \mathscr{E}_t(H^\mathsf{d}(\cdot,T_{n-1},T_n)),
\end{align}
where the process $H^\mathsf{d}(\cdot,T_{n-1},T_n)$ given by the stochastic logarithm
\begin{align*}
H^\mathsf{d}(\cdot,T_{n-1},T_n)=\mathscr{L}\left(\exp\left(\int_0^t \lambda^\mathsf{d}(s,T_{n-1})dL_s^{T_n} + \int_0^t b^\mathsf{d}(s,T_{n-1},T_n)ds \right) \right)
\end{align*}
is of the form
\begin{align} \label{form_sexp}
H^\mathsf{d}(t,T_{n-1},T_n) = & \int_0^t \lambda^\mathsf{d}(s,T_{n-1}) \sqrt{c_s} dW_s^{T^*} \nonumber \\
& + \int_0^t \int_{\R^d} \left( e^{\langle \lambda^\mathsf{d}(s,T_{n-1}),x \rangle} - 1 \right) (\mu^L - \nu^{T^*})(ds,dx).
\end{align}
Observe that by \citet[][Theorem I.4.61]{JacodShiryaev03} and \citet*{EberleinJacodRaible05} the forward price process as given in (\ref{discounting_forward_process_sexp}) is a true $P_{T^*}^\mathsf{d}$- martingale. 

Hence, we can specify the forward martingale measure associated with date $T_{n-1}$ defined on $(\Omega,\mathscr{F}_{T_{n-1}})$ and denoted by $P_{T_{n-1}}^\mathsf{d}$ by setting
\begin{align*}
\frac{dP_{T_{n-1}}^\mathsf{d}}{dP_{T_n}^\mathsf{d}} \coloneqq \frac{F^\mathsf{d}(T_{n-1},T_{n-1},T_n)}{F^\mathsf{d}(0,T_{n-1},T_n)} = \mathscr{E}_{T_{n-1}}(H^\mathsf{d}(\cdot,T_{n-1},T_n)). 
\end{align*}
Restricted to the $\sigma$-field $\mathscr{F}_t$ with $t \leq T_{n-1}$, we have
\begin{align*}
\frac{dP_{T_{n-1}}^\mathsf{d}\big|_{\mathscr{F}_t}}{dP_{T_n}^\mathsf{d}\big|_{\mathscr{F}_t}} = \frac{F^\mathsf{d}(t,T_{n-1},T_n)}{F^\mathsf{d}(0,T_{n-1},T_n)} = \mathscr{E}_t(H^\mathsf{d}(\cdot,T_{n-1},T_n)). 
\end{align*}

Applying Girsanov's theorem for semimartingales (see \citet[][Theorem III.3.24]{JacodShiryaev03}) we can identify the predictable processes $\beta$ and $Y$ that describe the change of measure from equation (\ref{form_sexp}). One obtains
\begin{align*}
\beta(t)=\lambda^\mathsf{d}(t,T_{n-1}) \quad \text{and} \quad Y(t,x)=\exp\left(\langle \lambda^\mathsf{d}(t,T_{n-1}), x \rangle \right).
\end{align*}
In particular, these processes determine the characteristics of the semimartingale $L^{T^*}$ relative to $P_{T_{n-1}}^\mathsf{d}$ from its semimartingale characteristics relative to $P_{T^*}^\mathsf{d}$.
Since the characteristics remain deterministic we conclude that $L^{T^*}$ remains a process with independent increments after the measure change. 

According to Girsanov's theorem the process $W^{T_{n-1}} \coloneqq (W_t^{T_{n-1}})_{t \in [0,T^*]}$ defined by
\begin{align*}
W_t^{T_{n-1}} \coloneqq W_t^{T^*} - \int_0^t \sqrt{c_s}\lambda^\mathsf{d}(s,T_{n-1})^\mathsf{T}ds
\end{align*}
is a $d$-dimensional standard Brownian motion under $P_{T_{n-1}}^\mathsf{d}$ and furthermore
\begin{align*}
\nu^{T_{n-1}}(dt,dx) \coloneqq \exp\left(\langle \lambda^\mathsf{d}(t,T_{n-1}), x \rangle \right) \nu^{T^*}(dt,dx) = F_t^{T_{n-1}}(dx)dt 
\end{align*}
defines the $P_{T_{n-1}}^\mathsf{d}$-compensator of $\mu^L$, where we set 
\begin{align*}
F_t^{T_{n-1}}(dx) \coloneqq \exp\left(\langle \lambda^\mathsf{d}(t,T_{n-1}), x \rangle \right) F_t^{T^*}(dx).
\end{align*}

Proceeding backwards along the discrete tenor structure $\mathscr{T}$, we get, for any $k \in \{1,\dots,n\}$ and $0 \leq t \leq T_{k-1}$, all forward price processes in the form
\begin{align*} 
\frac{F^\mathsf{d}(t,T_{k-1},T_k)}{F^\mathsf{d}(0,T_{k-1},T_k)} = \exp\Big( & \int_0^t \lambda^\mathsf{d}(s,T_{k-1})dL_s^{T_k} + \int_0^t b^\mathsf{d}(s,T_{k-1},T_k)ds \Big),
\end{align*}
where the process $L^{T_k}=(L_t^{T_k})_{t \in [0,T^*]}$ is given by
\begin{align} \label{original_canonical_representation}
L_t^{T_k} = \int_0^t \sqrt{c_s}dW_s^{T_k} + \int_0^t \int_{\R^d} x (\mu^L - \nu^{T_k})(ds,dx) 
\end{align}
and the drift $b^\mathsf{d}(\cdot,T_{k-1},T_k)$ is of the form
\begin{align*}
b^\mathsf{d}(t,T_{k-1},T_k) = & - \frac{1}{2} \langle \lambda^\mathsf{d}(t,T_{k-1}), c_t \lambda^\mathsf{d}(t,T_{k-1})^\mathsf{T} \rangle \\
& - \int_{\R^d} \left(e^{\langle \lambda^\mathsf{d}(t,T_{k-1}), x\rangle} - 1 - \langle \lambda^\mathsf{d}(t,T_{k-1}), x \rangle \right) F_t^{T_k}(dx).
\end{align*} 
The drift is specified in such a way that $F^\mathsf{d}(\cdot,T_{k-1},T_k)$ is a $P_{T_k}^\mathsf{d}$-martingale where in the respective previous step we have defined the forward measure $P_{T_k}^\mathsf{d}$ on $(\Omega,\mathscr{F}_{T^*},(\mathscr{F}_t)_{t \in [0,T_k]})$ by the density
\begin{align*}
\frac{dP_{T_k}^\mathsf{d}}{dP_{T_{k+1}}^\mathsf{d}}\Big|_{\mathscr{F}_t} \coloneqq \frac{dP_{T_k}^\mathsf{d}\big|_{\mathscr{F}_t}}{dP_{T_{k+1}}^\mathsf{d}\big|_{\mathscr{F}_t}} \coloneqq \frac{F^\mathsf{d}(t,T_k,T_{k+1})}{F^\mathsf{d}(0,T_k,T_{k+1})}
\end{align*}
which is related to $P_{T_l}^\mathsf{d}$ with $l \in \{k+1,\dots,n\}$ by
\begin{align*}
\frac{dP_{T_k}^\mathsf{d}}{dP_{T_l}^\mathsf{d}}\Big|_{\mathscr{F}_t} = \prod_{j=k}^{l-1} \frac{F^\mathsf{d}(t,T_j,T_{j+1})}{F^\mathsf{d}(0,T_j,T_{j+1})}=\frac{B_0^\mathsf{d}(T_l)}{B_0^\mathsf{d}(T_k)} \prod_{j=k}^{l-1} F^\mathsf{d}(t,T_j,T_{j+1}).
\end{align*}
Moreover, we get the relations
\begin{align*}
W_t^{T_k} \coloneqq W_t^{T^*} - \int_0^t \sqrt{c_s} \sum_{j=k}^{n-1} \lambda^\mathsf{d}(s,T_j)^\mathsf{T}ds
\end{align*}
and for the $P_{T_k}^\mathsf{d}$-compensator $\nu^{T_k}$ of $\mu^L$
\begin{align*}
\nu^{T_k}(dt,dx) \coloneqq & \exp\left(\sum_{j=k}^{n-1} \langle \lambda^\mathsf{d}(t,T_j), x \rangle \right) \nu^{T^*}(dt,dx) = F_t^{T_k}(dx)dt,
\end{align*}
where we have set $F_t^{T_k}(dx) \coloneqq \exp\left(\sum_{j=k}^{n-1} \langle \lambda^\mathsf{d}(t,T_j), x \rangle \right) F_t^{T^*}(dx)$. Note that the processes $L^{T_k}$ are time-inhomogeneous Lévy processes with local characteristics $(0,c_t,F_t^{T_k})$ under $P_{T_k}^\mathsf{d}$ that fulfil $(\mathbb{EM})$ due to assumption ($\mathbb{DFP}.2$). Formula (\ref{original_canonical_representation}) gives the canonical representation with respect to $P_{T_k}^\mathsf{d}$. 

To end this subsection, we derive for any $T,S \in \mathscr{T}=\{T_0,\dots,T_n\}$ with $T \leq S$, a general representation for the relationship between the driving processes $L^T$ and $L^S$. Let us denote
\begin{align} \label{definition_J_TS}
\mathscr{J}_T^S \coloneqq \{h \in \N|~T_h \in \mathscr{T} ~ \text{and}~ T < T_h \leq S\}
\end{align}
and define
\begin{align*}
w(s,T,S) \coloneqq & - c_s \sum_{h \in \mathscr{J}_T^S} \lambda^\mathsf{d}(s,T_{h-1})^\mathsf{T} \\
& + \int_{\R^d} x \Big[ \exp\Big(-\sum_{h \in \mathscr{J}_T^S} \langle \lambda^\mathsf{d}(s,T_{h-1}), x \rangle\Big) - 1 \Big] F_s^T(dx)
\end{align*}
for any $s \leq T$. Note that $w(s,T,S) \equiv 0$ when $T=S$. By an application of Girsanov's theorem for semimartingales, we obtain
\begin{align} \label{canonical_representation}
L^T = & - \int_0^{\cdot} c_s \sum_{h \in \mathscr{J}_T^S}\lambda^\mathsf{d}(s,T_{h-1})^{\mathsf{T}} ds \nonumber \\
& + \int_0^{\cdot} \int_{\R^d} x \Big[ \exp\Big(-\sum_{h \in \mathscr{J}_T^S} \langle \lambda^\mathsf{d}(s,T_{h-1}), x \rangle\Big) - 1 \Big] F_s^T(dx) ds \nonumber \\
& + \int_0^{\cdot} \sqrt{c_s}dW_s^S + \int_0^{\cdot} \int_{\R^d} x (\mu^L - \nu^S)(ds,dx) \nonumber \\
= & \int_0^{\cdot} w(s,T,S)ds + L^S. 
\end{align}

We emphasize that all driving processes $L^T$ remain time-inhomogeneous Lévy processes under the corresponding forward measures since they differ only by deterministic drift terms.
 
\subsection{Risky tenor-dependent curves} \label{multiple_term_structure_curves}

Let $m \in \N$ be the number of curves. For every $i \in \{1,\dots,m\}$, we consider an equidistant discrete tenor structure $\mathscr{T}^i \coloneqq \{T_0^i,\dots,T_{n_i}^i\}$ corresponding to curve $i$, where $n_i, n \in \N$, $\mathscr{T}^i \subset \mathscr{T}=\{T_0,\dots,T_n\}$ and $0 \leq T_0^i=T_0 < T_1^i < \dots < T_{n_i}^i=T_n=T^*$. As before, the year fractions between the dates $T_{k-1}^i$ and $T_k^i$ are denoted by $\delta^i \coloneqq \delta^i(T_{k-1}^i,T_k^i)$ for all $k \in \{1,\dots,n_i\}$. Furthermore, we postulate that $\mathscr{T}^m \subset \dots \subset \mathscr{T}^1 \subset \mathscr{T}$. 


We consider the time-inhomogeneous Lévy process $L^{T^*}$ on $(\Omega,\mathscr{F}_{T^*},\mathbb{F}=(\mathscr{F}_t)_{t \in [0,T^*]},P_{T^*}^\mathsf{d})$ and probability measures $P_{T_1}^\mathsf{d},\dots,P_{T_{n-1}}^\mathsf{d}$ from the last subsection. Observe that we have
\begin{align} \label{forward_price_Tki}
F^\mathsf{d}(t,T_{k-1}^i,T_k^i) = \prod_{j \in \mathscr{J}_k^i} F^\mathsf{d}(t,T_{j-1},T_j),
\end{align}
where $\mathscr{J}_k^i$ is a short form of $\mathscr{J}_{T_{k-1}^i}^{T_k^i}$ that is defined in (\ref{definition_J_TS}). 

For each $i \in \{1,\dots,m\}$ and $k \in \{1,\dots,n_i\}$, let $L^i(T_{k-1}^i,T_k^i)$ denote the $T_{k-1}^i$-spot Libor or Euribor rate corresponding to tenor $\delta^i$. Let us assume that $L^i(T_{k-1}^i,T_k^i)$ is an $\mathscr{F}_{T_{k-1}^i}$-measurable random variable. Then, we define
\begin{align} \label{reference_rate}
L^i(t,T_{k-1}^i,T_k^i) \coloneqq \E_{T_k^i}^\mathsf{d}\big[L^i(T_{k-1}^i,T_k^i)|\mathscr{F}_t\big] \quad \quad \quad (t \leq T_{k-1}^i)
\end{align}
where $\E_{T_k^i}^\mathsf{d}[\cdot]$ denotes the expectation under $P_{T_k^i}^\mathsf{d}$. Notice that this definition corresponds to the valuation formula of the market rate of a (textbook) forward rate agreement \cite[see][]{Mercurio09} and we have 
\begin{align*}
L^i(T_{k-1}^i,T_{k-1}^i,T_k^i) = L^i(T_{k-1}^i,T_k^i).
\end{align*} 
Furthermore, $L^i(\cdot,T_{k-1}^i,T_k^i)$ is, by definition, a $P_{T_k^i}^\mathsf{d}$-martingale. We refer to $L^i(t,T_{k-1}^i,T_k^i)$ as the discretely compounded forward rate corresponding to $\delta^i$. These rates are the key quantities of the model. 

\begin{lemma} \label{lemma_discount_forward_process}
An explicit representation of the forward price $F^\mathsf{d}(t,T_{k-1}^i,T_k^i)$ in tenor structure $\mathscr{T}^i$ is given by
\begin{align*}
& F^\mathsf{d}(t,T_{k-1}^i,T_k^i) = F^\mathsf{d}(0,T_{k-1}^i,T_k^i) \exp \Bigg( \int_0^t \sum_{j \in \mathscr{J}_k^i} \lambda^\mathsf{d}(s,T_{j-1}) dL_s^{T_k^i} \\
& + \int_0^t \sum_{j \in \mathscr{J}_k^i} \big[ \langle \lambda^\mathsf{d}(s,T_{j-1}), w(s,T_j,T_k^i)\rangle + b^\mathsf{d}(s,T_{j-1},T_j) \big] ds \Bigg)
\end{align*}
and $F^\mathsf{d}(\cdot,T_{k-1}^i,T_k^i)$ is a $P_{T_k^i}^\mathsf{d}$-martingale.
\end{lemma}
\begin{proof}
Considering representation (\ref{forward_price_Tki}) and applying equation (\ref{canonical_representation}), we obtain this explicit formula for $F^\mathsf{d}(t,T_{k-1}^i,T_k^i)$. Applying \citet[][Proposition III.3.8]{JacodShiryaev03} successively to the product in (\ref{forward_price_Tki}), the martingale property is proved.
\end{proof}

Instead of modelling the dynamics of the forward rates $L^i$ directly, we specify the evolution by modelling the forward spreads relative to $L^\mathsf{d}$ and $F^\mathsf{d}$. Typically, these spreads can be considered in two ways, namely as 
\begin{enumerate}
\item additive forward spreads
\begin{align*}
s^i(t,T_{k-1}^i,T_k^i) \coloneqq L^i(t,T_{k-1}^i,T_k^i) - L^\mathsf{d}(t,T_{k-1}^i,T_k^i)
\end{align*}
or
\item multiplicative forward spreads 
\begin{align*}
S^i(t,T_{k-1}^i,T_k^i) \coloneqq \frac{1 + \delta^i L^i(t,T_{k-1}^i,T_k^i)}{1 + \delta^i L^\mathsf{d}(t,T_{k-1}^i,T_k^i)}= \frac{1 + \delta^i L^i(t,T_{k-1}^i,T_k^i)}{F^\mathsf{d}(t,T_{k-1}^i,T_k^i)}.
\end{align*}
\end{enumerate}
The natural choice in the forward price framework are multiplicative forward spreads. Following this approach we can easily ensure the observed monotonicity between each risky curve and the basic one by modelling the multiplicative forward spreads as quantities which are larger than one.

We have
\begin{align} \label{relation_spreads}
1 + \delta^i L^i(\cdot,T_{k-1}^i,T_k^i) = S^i(\cdot,T_{k-1}^i,T_k^i)F^\mathsf{d}(\cdot,T_{k-1}^i,T_k^i).
\end{align}
\begin{lemma}
For each pair of dates $T_{k-1}^i, T_k^i \in \mathscr{T}^i$ with $i \in \{1,\dots,m\}$ and $k \in \{1,\dots,n_i\}$, the process $L^i(\cdot,T_{k-1}^i,T_k^i)$ follows a $P_{T_k^i}^\mathsf{d}$-martingale if and only if the multiplicative forward spread $S^i(\cdot,T_{k-1}^i,T_k^i)$ is a $P_{T_{k-1}^i}^\mathsf{d}$-martingale.
\end{lemma}
\begin{proof}
Proposition III.3.8 in \citet[][]{JacodShiryaev03} can be directly applied together with
\begin{align*}
\frac{dP_{T_{k-1}^i}^\mathsf{d}}{dP_{T_k^i}^\mathsf{d}}\Big|_{\mathscr{F}_t} = \frac{F^\mathsf{d}(t,T_{k-1}^i,T_k^i)}{F^\mathsf{d}(0,T_{k-1}^i,T_k^i)}.
\end{align*} 
\end{proof}
Consequently, recalling that $L^i(\cdot,T_{k-1}^i,T_k^i)$ by definition (\ref{reference_rate}) is a $P_{T_k^i}^\mathsf{d}$-martingale, the evolution of $S^i(\cdot,T_{k-1}^i,T_k^i)$ has to be specified in such a way that it is a $P_{T_{k-1}^i}^\mathsf{d}$-martingale. 

For the starting values of the forward rate, we set
\begin{align*}
L^i(0,T_{k-1}^i,T_k^i) =  \E_{T_k^i}^\mathsf{d}\big[L^i(T_{k-1}^i,T_k^i)\big] = \mathsf{FRA}(0,T_{k-1}^i,T_k^i)
\end{align*}
where $\mathsf{FRA}(0,T_{k-1}^i,T_k^i)$ denotes the current market rate of a textbook forward rate agreement with respect to dates $T_{k-1}^i$ and $T_k^i$. The initial values of the forward spreads are then given by
\begin{align} \label{initial_spreads}
S^i(0,T_{k-1}^i,T_k^i) = \frac{1 + \delta^i L^i(0,T_{k-1}^i,T_k^i)}{F^\mathsf{d}(0,T_{k-1}^i,T_k^i)},
\end{align}
where $F^\mathsf{d}(0,T_{k-1}^i,T_k^i)$ is also obtained from market data combining (\ref{discount_forward_process}) and (\ref{forward_price_Tki}). More exactly, the quantities on the right hand side of (\ref{initial_spreads}) are obtained from the multiple term structure curves that are bootstrapped by using data from tenor specific deposits, forward rate agreements and swaps.  

Below, we shall develop two modelling approaches offering different levels of tractability in exchange for the properties which one wants to achieve in the model. 

\subsubsection{Model (a)}

To get a maximum of tractability we allow in the first approach that the value of the multiplicative forward spreads may become less than one which is equivalent to additive spreads being negative. To this end, we model $S^i(\cdot,T_{k-1}^i,T_k^i)$ below as an ordinary exponential. 

We start with the following input which can be interpreted as the volatility of $S^i(\cdot,T_{k-1}^i,T_k^i)$.
\begin{enumerate}
\item[($\mathbb{MFP}.a$)] For each $i \in \{1,\dots,m\}$ and each maturity $T_{k-1}^i \in \mathscr{T}^i$ with $k \in \{1,\dots,n_i\}$, there is a bounded, continuous and deterministic function $\gamma^i(\cdot,T_{k-1}^i)$ given by
\begin{align*}
\gamma^i(\cdot,T_{k-1}^i):
\begin{cases}
[0,T^*] \rightarrow \R_+^d \\
t \mapsto \gamma^i(t,T_{k-1}^i)=(\gamma^{i,1}(t,T_{k-1}^i),\dots,\gamma^{i,d}(t,T_{k-1}^i))
\end{cases}
\end{align*}
where we require 
\begin{align*}
\sum_{k=1}^{n_i} \left( \lambda^{\mathsf{d},j}(t,T_{k-1}^i) + \gamma^{i,j}(t,T_{k-1}^i) \right) \leq M, \quad \text{for all}~ t \in [0,T^*]~\text{and}~j \in \{1,\dots,d\}.
\end{align*}
Again the constant $M$ is from assumption $(\mathbb{EM})$ and we set $\gamma^i(t,T_{k-1}^i)=(0,\dots,0)$ for $t > T_{k-1}^i$.
\end{enumerate}

We postulate for any $i \in \{1,\dots,m\}$ and each pair of dates $T_{k-1}^i, T_k^i$ that
\begin{align*}
\frac{S^i(t,T_{k-1}^i,T_k^i)}{S^i(0,T_{k-1}^i,T_k^i)} = \exp\Big( \int_0^t \gamma^i(s,T_{k-1}^i)dL_s^{T_{k-1}^i} + \int_0^t b^i(s,T_{k-1}^i)ds \Big)
\end{align*}
where $L^{T_{k-1}^i}$ is defined in subsection \ref{basic_curve} and the drift term $b^i(\cdot,T_{k-1}^i)$ is chosen such that $S^i(\cdot,T_{k-1}^i,T_k^i)$ is a $P_{T_{k-1}^i}^\mathsf{d}$-martingale, namely
\begin{align*}
b^i(t,T_{k-1}^i) = & - \frac{1}{2} \langle \gamma^i(t,T_{k-1}^i), c_t \gamma^i(t,T_{k-1}^i)^\mathsf{T} \rangle \\
& - \int_{\R^d} \left( e^{\langle \gamma^i(t,T_{k-1}^i), x\rangle} - 1 - \langle \gamma^i(t,T_{k-1}^i), x \rangle \right) F_t^{T_{k-1}^i}(dx).
\end{align*}
By Lemma \ref{lemma_discount_forward_process}, we have
\begin{align*}
\frac{F^\mathsf{d}(t,T_{k-1}^i,T_k^i)}{ F^\mathsf{d}(0,T_{k-1}^i,T_k^i) } & = \exp\Bigg( \int_0^t \sum_{j \in \mathscr{J}_k^i} \lambda^\mathsf{d}(s,T_{j-1}) dL_s^{T_k^i} \\
& + \int_0^t \sum_{j \in \mathscr{J}_k^i} \big[ \langle \lambda^\mathsf{d}(s,T_{j-1}), w(s,T_j,T_k^i) \rangle + b^\mathsf{d}(s,T_{j-1},T_j) \big]ds \Bigg)
\end{align*}
and $F^\mathsf{d}(\cdot,T_{k-1}^i,T_k^i)$ is a $P_{T_k^i}^\mathsf{d}$-martingale. This forward price process represents at the same time the density process for the following measure change, namely 
\begin{align*}
\frac{dP_{T_{k-1}^i}^\mathsf{d}}{dP_{T_k^i}^\mathsf{d}}\Big|_{\mathscr{F}_t} = \frac{F^\mathsf{d}(t,T_{k-1}^i,T_k^i)}{F^\mathsf{d}(0,T_{k-1}^i,T_k^i)}.
\end{align*} 

Using representation (\ref{relation_spreads}) and applying (\ref{canonical_representation}), we obtain the forward rate $L^i(\cdot,T_{k-1}^i,T_k^i)$ from
\begin{align*}
& 1 + \delta^i L^i(t,T_{k-1}^i,T_k^i) = S^i(t,T_{k-1}^i,T_k^i) F^\mathsf{d}(t,T_{k-1}^i,T_k^i)\\
= & \left( 1 + \delta^i L^i(0,T_{k-1}^i,T_k^i) \right)  \exp \Bigg( \int_0^t \gamma^i(s,T_{k-1}^i)dL_s^{T_{k-1}^i} + \int_0^t b^i(s,T_{k-1}^i)ds \Bigg) \\
& \times \exp \Bigg( \int_0^t \sum_{j \in \mathscr{J}_k^i} \lambda^\mathsf{d}(s,T_{j-1}) dL_s^{T_k^i} \\
& + \int_0^t \sum_{j \in \mathscr{J}_k^i} \big[ \langle \lambda^\mathsf{d}(s,T_{j-1}), w(s,T_j,T_k^i) \rangle + b^\mathsf{d}(s,T_{j-1},T_j) \big] ds \Bigg) \\
= & \left( 1 + \delta^i L^i(0,T_{k-1}^i,T_k^i) \right) 
\exp \Bigg( \int_0^t \big[ \sum_{j \in \mathscr{J}_k^i} \lambda^\mathsf{d}(s,T_{j-1}) + \gamma^i(s,T_{k-1}^i)\big] dL_s^{T_k^i} \\
& + \int_0^t \Big[ b^i(s,T_{k-1}^i) + \langle \gamma^i(s,T_{k-1}^i), w(s,T_{k-1}^i,T_k^i) \rangle \\
& + \sum_{j \in \mathscr{J}_k^i} \big[ \langle \lambda^\mathsf{d}(s,T_{j-1}), w(s,T_j,T_k^i) \rangle + b^\mathsf{d}(s,T_{j-1},T_j) \big] \Big] ds \Bigg).
\end{align*} 

\textbf{Remark:} In this model, the forward reference rates as well as the $\delta^i$-forward rates can become negative in accordance with the current market situation. In particular, the initial rates can already be negative. Note that starting from positive (negative) initial rates does not mean that the rates remain positive (negative) over time.  
 
\subsubsection{Model (b)}

Now we will specify a model which ensures that the multiplicative forward spreads are larger than one if the initial spreads are already larger than one. This is equivalent to the positivity of the additive spreads. We mention already at this point that the pricing of derivatives becomes only slightly less tractable than in model (a). 

As in model (a) we need the volatility functions to satisfy the following conditions
\begin{enumerate}
\item[($\mathbb{MFP}.b$)] For each $i \in \{1,\dots,m\}$ and each maturity $T_{k-1}^i \in \mathscr{T}^i$ with $k \in \{1,\dots,n_i\}$, there is a bounded, continuous and deterministic function $\bar{\gamma}^i(\cdot,T_{k-1}^i)$ given by
\begin{align*}
\bar{\gamma}^i(\cdot,T_{k-1}^i):
\begin{cases}
[0,T^*] \rightarrow \R_+^d \\
t \mapsto \bar{\gamma}^i(t,T_{k-1}^i)=(\bar{\gamma}^{i,1}(t,T_{k-1}^i),\dots,\bar{\gamma}^{i,d}(t,T_{k-1}^i))
\end{cases}
\end{align*}
where we require 
\begin{align*}
\sum_{k=1}^{n_i} \left( \lambda^{\mathsf{d},j}(t,T_{k-1}^i) + \bar{\gamma}^{i,j}(t,T_{k-1}^i) \right) \leq M, \quad \text{for all}~ t \in [0,T^*]~\text{and}~j \in \{1,\dots,d\}.
\end{align*}
We set $\bar{\gamma}^i(t,T_{k-1}^i)=(0,\dots,0)$ for $t > T_{k-1}^i$.
\end{enumerate} 

For any $i \in \{1,\dots,m\}$ and each pair of dates $T_{k-1}^i, T_k^i$ we assume that 
\begin{align} \label{spreads_b}
\frac{S^i(t,T_{k-1}^i,T_k^i) - 1}{S^i(0,T_{k-1}^i,T_k^i) - 1} =   \exp\Big(  \int_0^t \bar{\gamma}^i(s,T_{k-1}^i)dL_s^{T_{k-1}^i} + \int_0^t \bar{b}^i(s,T_{k-1}^i)ds \Big),
\end{align}
where 
\begin{align} \label{exp_comp}
\bar{b}^i(t,T_{k-1}^i) = & - \frac{1}{2} \langle \bar{\gamma}^i(t,T_{k-1}^i), c_t \bar{\gamma}^i(t,T_{k-1}^i)^\mathsf{T} \rangle  \nonumber \\
& - \int_{\R^d} \left( e^{\langle \bar{\gamma}^i(t,T_{k-1}^i), x \rangle} - 1 - \langle \bar{\gamma}^i(t,T_{k-1}^i), x \rangle \right) F_t^{T_{k-1}^i}(dx).
\end{align}
Note that 
\begin{align*} 
& S^i(t,T_{k-1}^i,T_k^i)  =  \\
& 1 + (S^i(0,T_{k-1}^i,T_k^i) -1)  \exp\Big(  \int_0^t \bar{\gamma}^i(s,T_{k-1}^i)dL_s^{T_{k-1}^i} + \int_0^t \bar{b}^i(s,T_{k-1}^i)ds \Big)
\end{align*}
and one sees that $S^i(\cdot,T_{k-1}^i,T_k^i)$ is a $P_{T_{k-1}^i}^\mathsf{d}$-martingale by the choice of the exponential compensator defined via (\ref{exp_comp}).

In an analogous way as in the previous subsection, we get a representation for $1+\delta^iL^i(t,T_{k-1}^i,T_k^i)$, namely
\begin{align*}
& 1 + \delta^i L^i(t,T_{k-1}^i,T_k^i) = D^\mathsf{d}(t,T_{k-1}^i,T_k^i) \exp \Bigg( \int_0^t \sum_{j \in \mathscr{J}_k^i} \lambda^\mathsf{d}(s,T_{j-1}) dL_s^{T_k^i} \Bigg) \nonumber \\
& + D^i(t,T_{k-1}^i,T_k^i) \exp \Bigg( \int_0^t \Big[\sum_{j \in \mathscr{J}_k^i} \lambda^\mathsf{d}(s,T_{j-1}) + \bar{\gamma}^i(s,T_{k-1}^i)\Big]dL_s^{T_k^i} \Bigg),
\end{align*}
where we set
\begin{align*}
D^\mathsf{d}(t,T_{k-1}^i,T_k^i) \coloneqq & F^\mathsf{d}(0,T_{k-1}^i,T_k^i) \exp \Bigg( \int_0^t \sum_{j \in \mathscr{J}_k^i} \big[  b^\mathsf{d}(s,T_{j-1},T_j) \\
& +  \langle \lambda^\mathsf{d}(s,T_{j-1}), w(s,T_j,T_k^i) \rangle \big] ds \Bigg)
\end{align*}
and
\begin{align*}
D^i(t,T_{k-1}^i,T_k^i) \coloneqq F^\mathsf{d}(0,T_{k-1}^i,T_k^i)(S^i(0,T_{k-1}^i,T_k^i) - 1) \exp\Bigg( \int_0^t d^i(s,T_{k-1}^i,T_k^i) ds \Bigg)
\end{align*}
with
\begin{align*}
d^i(s,T_{k-1}^i,T_k^i) \coloneqq & \sum_{j \in \mathscr{J}_k^i} \big[ \langle \lambda^\mathsf{d}(s,T_{j-1}), w(s,T_j,T_k^i)\rangle + b^\mathsf{d}(s,T_{j-1},T_j) \big] \\
& + \bar{b}^i(s,T_{k-1}^i) + \langle \bar{\gamma}^i(s,T_{k-1}^i), w(s,T_{k-1}^i,T_k^i) \rangle.
\end{align*}
Note that for the initial values we have the relation
\begin{align*}
& F^\mathsf{d}(0,T_{k-1}^i,T_k^i)(S^i(0,T_{k-1}^i,T_k^i) - 1) = \\
& \quad \quad \quad \quad 1+ \delta^i L^i(0,T_{k-1}^i,T_k^i) - \prod_{j \in \mathscr{J}_k^i} [1 + \delta L^\mathsf{d}(0,T_{j-1},T_j)].
\end{align*}

\section{Pricing formula for caps}

Let $l \in \{1,\dots,m \}$. The time-$t$ price of a caplet with tenor $\delta^l$,  maturity $T \in \mathscr{T}^l$ and strike $K$, for $t \leq T$ and $T + \delta^l = T_k \in \mathscr{T}^l$ for some $k \in \{1,\dots,n\}$, is given by
\begin{align} \label{valuation_caplet}
\mathsf{Cpl}(t,T,\delta^l,K) & \coloneqq \delta^l B_t^\mathsf{d}(T_k) \E_{T_k}^\mathsf{d}\big[\left(L^l(T,T_k) - K \right)^+ \big| \mathscr{F}_t \big] \nonumber \\
& = B_t^\mathsf{d}(T_k) \E_{T_k}^\mathsf{d}\big[\left(1 + \delta^l L^l(T,T,T_k) - (1 + \delta^l K) \right)^+ \big| \mathscr{F}_t \big] \nonumber \\
& = B_t^\mathsf{d}(T_k) \E_{T_k}^\mathsf{d}\big[\Big( F^\mathsf{d}(T,T,T_k) S^l(T,T,T_k) - \tilde{K}^l \Big)^+ \big| \mathscr{F}_t \big] \nonumber \\
& = B_t^\mathsf{d}(T_k) (Z_t^k)^{-1} \E_{T^*}^\mathsf{d}\big[Z_T^k \Big(  F^\mathsf{d}(T,T,T_k) S^l(T,T,T_k) - \tilde{K}^l \Big)^+ \big| \mathscr{F}_t \big]
\end{align}
where $\tilde{K}^l \coloneqq 1 + \delta^l K$ and
\begin{align*}
Z_T^k \coloneqq 
\begin{cases}
\prod_{j \in \mathscr{J}_{T_k}^{T^*}} \frac{F^\mathsf{d}(T,T_{j-1},T_j)}{F^\mathsf{d}(0,T_{j-1},T_j)}, &\quad T_k < T^* \\
1, &\quad T_k = T^*.
\end{cases}
\end{align*}

We make the following 
\begin{vol} \label{VOL}
For every $l \in \{1,\dots,m\}$ and all $T \in [0,T^*]$, the volatility functions are decomposable in the form
\begin{align*} 
\lambda^\mathsf{d}(t,T) = \lambda_1^\mathsf{d}(T) \lambda(t)
\end{align*}
\begin{align*}
\gamma^l(t,T) = \gamma_1^l(T) \lambda(t)
\end{align*}
and
\begin{align*}
\bar{\gamma}^l(t,T) = \bar{\gamma}_1^l(T) \lambda(t)
\end{align*}
where $\lambda_1^\mathsf{d}$, $\gamma_1^l$ and $\bar{\gamma}_1^l:[0,T^*] \rightarrow \R_+$ and $\lambda:[0,T^*] \rightarrow \R^d$ are deterministic and continuous functions. The vector $\lambda(t)$ is bounded in the sense of
\begin{align*}
\abs{\lambda^k(t)} \leq M^{'}
\end{align*}
for every $t \in [0,T^*]$, $k \in \{1,\dots,d\}$ and a constant $M^{'}<M$. Furthermore, we assume that
\begin{align*}
\Lambda^l(T,T_n) \coloneqq \sum_{j \in \mathscr{J}_T^{T_n}} \lambda_1^\mathsf{d}(T_{j-1}) + \gamma_1^l(T) < R
\end{align*}
for $R=1+\frac{M-M^{'}}{M^{'}}$ and all $T \in \mathscr{T}$. In the same way we assume
\begin{align*}
\bar{\Lambda}^l(T,T_n) \coloneqq \sum_{j \in \mathscr{J}_T^{T_n}} \lambda_1^\mathsf{d}(T_{j-1}) + \bar{\gamma}_1^l(T) < R
\end{align*}
for all $T \in \mathscr{T}$.
\end{vol}

Let us define the random variable $X_T \coloneqq \int_0^T \lambda(s)dL_s^{T^*}$ and consider its extended characteristic function $\varphi_{X_T}$ under $P_{T^*}^\mathsf{d}$ which can be expressed as
\begin{align} \label{char_X}
\varphi_{X_T}(z) = \exp \left( \int_0^T \theta_s\left(iz \lambda(s)\right)ds \right).
\end{align}
For details of this representation compare \citet[Lemma 3.1]{EberleinRaible99}. 

Recall that the extended cumulant of $L^{T^*}$ with respect to $P_{T^*}^\mathsf{d}$ is given by
\begin{align*}
\theta_t(z) = \frac{1}{2} \langle z, c_t z \rangle + \int_{\R^d} \left(e^{\langle z,x\rangle}-1-\langle z,x \rangle \right) F_t^{T^*}(dx)
\end{align*}
for any $z \in \C$ where it is defined.

We will apply the Fourier based valuation method in order to make the formula for the time-0 price of a caplet numerically accessible. We write the right-hand side of (\ref{valuation_caplet}) for $t=0$ as an expected value of a payoff function $f_K^{k,l}$ applied to $X_T$. Since $Z_0^k \equiv 1$  we obtain
\begin{align} \label{caplet_price_formula}
\mathsf{Cpl}(0,T,\delta^l,K) = B_0^\mathsf{d}(T_k) \E_{T^*}^\mathsf{d}\big[ f_K^{k,l}(X_T) \big].
\end{align}
The explicit form of the function $f_K^{k,l}$ will be derived with the help of (\ref{canonical_representation}), the respective form of the expression 
\begin{align*}
1 + \delta^l L^l(T,T,T_k)=F^\mathsf{d}(T,T,T_k) S^l(T,T,T_k) 
\end{align*}
and the density (for $k \leq n-1$)
\begin{align*}
Z_T^k = \exp\Bigg( & \int_0^T \sum_{j \in \mathscr{J}_{T_k}^{T^*}} \lambda^\mathsf{d}(s,T_{j-1})dL_s^{T^*} + \int_0^T \big[ \sum_{j \in \mathscr{J}_{T_k}^{T^*}} b^\mathsf{d}(s,T_{j-1},T_j) \\
& + \mathds{1}_{\{k \leq n-2\}}  \sum_{j \in \mathscr{J}_{T_k}^{T_{n-1}}}  \langle \lambda^\mathsf{d}(s,T_{j-1}), w(s,T_j,T^*) \rangle \big]ds\Bigg).
\end{align*}

\subsection{Numerics}

In this and the following subsection we derive  numerically efficient forms of the caplet price formula (\ref{caplet_price_formula}) for the two model variants.

\subsubsection{Model (a)}

For any $k \in \{1,\dots,n-1\}$, we have
\begin{align*}
f_K^{k,l}(x) \coloneqq \Bigg( & \bar{D}^l(T,T,T_k) \exp (\big[ \sum_{j \in \mathscr{J}_T^{T_n}} \lambda_1^\mathsf{d}(T_{j-1}) + \gamma_1^l(T)\big] x) \\
&  - \tilde{K}^l A^\mathsf{d}(T,T_k) \exp \Big( \sum_{j \in \mathscr{J}_{T_k}^{T_n}} \lambda_1^\mathsf{d}(T_{j-1}) x \Big) \Bigg)^+
\end{align*}
and
\begin{align*}
f_K^{n,l}(x) \coloneqq \Bigg( & \hat{D}^l(T,T,T_n) \exp (\big[ \sum_{j \in \mathscr{J}_T^{T_n}} \lambda_1^\mathsf{d}(T_{j-1}) + \gamma_1^l(T)\big] x) - \tilde{K}^l  \Bigg)^+
\end{align*}
where we set
\begin{align*}
\bar{D}^l(T,T,T_k) \coloneqq \hat{D}^l(T,T,T_k) A^\mathsf{d}(T,T_k) \hat{C}^l(T,T_k)
\end{align*}
with
\begin{align*}
 \hat{D}^l(T,T,T_k) \coloneqq & \left( 1 + \delta^l L^l(0,T,T_k) \right) \\
& \exp \Bigg(  \int_0^T \Big[ b^l(s,T) + \langle \gamma^l(s,T), w(s,T,T_k) \rangle \\
& + \sum_{j \in \mathscr{J}_T^{T_k}} \big[ \langle \lambda^\mathsf{d}(s,T_{j-1}), w(s,T_j,T_k) \rangle + b^\mathsf{d}(s,T_{j-1},T_j) \big] \Big] ds \Bigg)
\end{align*}
\begin{align*}
A^\mathsf{d}(T,T_k) \coloneqq \exp\Bigg( & \int_0^T \big[ \sum_{j \in \mathscr{J}_{T_k}^{T^*}} b^\mathsf{d}(s,T_{j-1},T_j) \\
& + \mathds{1}_{\{k \leq n-2\}}  \sum_{j \in \mathscr{J}_{T_k}^{T_{n-1}}} \langle \lambda^\mathsf{d}(s,T_{j-1}), w(s,T_j,T^*) \rangle \big]ds\Bigg)
\end{align*}
and
\begin{align*}
\hat{C}^l(T,T_k) \coloneqq \exp\Bigg( \int_0^T  \langle \sum_{j \in \mathscr{J}_T^{T_k}} \lambda^\mathsf{d}(s,T_{j-1}) + \gamma^l(s,T), w(s,T_k,T^*) \rangle ds \Bigg).
\end{align*}

The dampened payoff function  is defined by
\begin{align*}
g_K^{k,l}(x) \coloneqq e^{-R x} f_K^{k,l}(x)
\end{align*}
for any $x \in \R$ and some $R \in \R$.

\begin{proposition} \label{prop1}
The time-0 price of the caplet is given by
\begin{align} \label{fourier_caplet_priceA}
\mathsf{Cpl}(0,T,\delta^l,K) = \frac{B_0^\mathsf{d}(T_k)}{\pi} \int_0^\infty \mathsf{Re}\left(\varphi_{X_T}(u-iR) \hat{f}_K^{k,l}(iR-u)\right) du
\end{align}
where $\hat{f}_K^{k,l}$ denotes the extended Fourier transform of $f_K^{k,l}$ admitting the representation
\begin{align*}
\hat{f}_K^{k,l}(z) = e^{i z x_k} \Big[ & - \bar{D}^l(T,T,T_k) \frac{e^{\big[\sum_{j \in \mathscr{J}_T^{T_n}} \lambda_1^\mathsf{d}(T_{j-1}) + \gamma_1^l(T) \big] x_k}}{\sum_{j \in \mathscr{J}_T^{T_n}} \lambda_1^\mathsf{d}(T_{j-1}) + \gamma_1^l(T)  + iz} \\
& + \tilde{K}^l A^\mathsf{d}(T,T_k) \frac{e^{\sum_{j \in \mathscr{J}_{T_k}^{T_n}} \lambda_1^\mathsf{d}(T_{j-1}) x_k}}{\sum_{j \in \mathscr{J}_{T_k}^{T_n}} \lambda_1^\mathsf{d}(T_{j-1}) + iz} \Big]
\end{align*}
for any $k \in \{1,\dots,n-1\}$ and
\begin{align*}
\hat{f}_K^{n,l}(z) = e^{i z x_n} \Big[ & - \hat{D}^l(T,T,T_n) \frac{e^{\big[ \sum_{j \in \mathscr{J}_T^{T_n}} \lambda_1^\mathsf{d}(T_{j-1}) + \gamma_1^l(T)\big] x_n}}{ \sum_{j \in \mathscr{J}_T^{T_n}} \lambda_1^\mathsf{d}(T_{j-1}) + \gamma_1^l(T) + iz} + \frac{\tilde{K}^l}{iz} \Big],
\end{align*}
where $z \in \C$ with 
\begin{align*}
\sum_{j \in \mathscr{J}_T^{T_n}} \lambda_1^\mathsf{d}(T_{j-1}) + \gamma_1^l(T) < \mathsf{Im}(z) 
\end{align*}
for all $T \in \mathscr{T}$, $x_k$ is the unique root of the function
\begin{align*}
h^{k,l}(x) \coloneqq & \hat{D}^l(T,T,T_k) \hat{C}^l(T,T_k) \exp \Big(\big[ \sum_{j \in \mathscr{J}_T^{T_k}} \lambda_1^\mathsf{d}(T_{j-1}) + \gamma_1^l(T)\big] x\Big)  - \tilde{K}^l,  
\end{align*}
$x_n$ is the unique root of the function
\begin{align*}
h^{n,l}(x) \coloneqq &  \hat{D}^l(T,T,T_n) \exp \Big(\big[ \sum_{j \in \mathscr{J}_T^{T_n}} \lambda_1^\mathsf{d}(T_{j-1}) + \gamma_1^l(T)\big] x\Big) - \tilde{K}^l 
\end{align*}
and $R=1+\frac{M-M^{'}}{M^{'}}$.
\end{proposition}
\begin{proof}
The explicit form of the Fourier transforms $\hat{f}_K^{k,l}$ and $\hat{f}_K^{n,l}$ follows from a simple integration exercise. In order to get the caplet formula we apply Theorem 2.2 of \citet{EberleinGlauPapapantoleon10}. Consequently, conditions $\textbf{(C1)}: g_K^{k,l} \in L_{\mathsf{bc}}^1(\R)$, $\textbf{(C2)}: M_{X_T}(R) < \infty$ and $\textbf{(C3)}: \hat{g}_K^{k,l} \in L^1(\R)$ of this theorem have to be verified.

The functions $h^{k,l}$ are strictly increasing and continuous with varying sign and therefore possess a unique root $x_k\in \R$ for every $k \in \{1,\dots,n\}$. Using assumption $(\mathbb{EM})$, the explicit form of $M_{X_T}$ derived from (\ref{char_X}) and the boundedness of $\lambda$ one can find an $R \in (0,1 + \frac{M-M^{'}}{M^{'}}]$ such that condition $\textbf{(C2)}$ is satisfied. We chose $R = 1 + \frac{M-M^{'}}{M^{'}}$ in assumption $(\mathbb{VOL})$. Recall 
\begin{align*} 
\Lambda^l(T,T_n) < R
\end{align*}
for all $T \in \mathscr{T}$. The dampened functions $g_K^{k,l}$ are obviously continuous and for $k \in \{1,\dots,n-1\}$ bounded by
\begin{align*}
g_K^{k,l}(x) \leq \bar{D}^l(T,T,T_k) e^{(\Lambda^l(T,T_n)-R) x_k}. 
\end{align*}
Therefore these functions are also integrable. For $k=n$ we have to replace $\bar{D}^l(T,T,T_k)$ by $\hat{D}^l(T,T,T_n) $. To verify condition $\textbf{(C3)}$, we use Lemma 2.5 of \citet*[][]{EberleinGlauPapapantoleon10}. Let us consider the Sobolev space
\begin{align*}
H^1(\R) \coloneqq \{g \in L^2(\R)| \partial g~\text{exists and}~\partial g \in L^2(\R)\}
\end{align*}
where $\partial g$ denotes the weak derivative of the function $g$. Due to this Lemma it suffices to show that $g_K^{k,l} \in H^1(\R)$, but this is clear because of the form of the function and the upper bounds given above. Hence Theorem 2.2. in the mentioned paper implies
\begin{align*}
\mathsf{Cpl}(0,T,\delta^l,K) = \frac{B_0^\mathsf{d}(T_k)}{2 \pi} \int_\R \varphi_{X_T}(u-iR) \hat{f}_K^{k,l}(iR-u) du.
\end{align*}
An obvious symmetry property of the integrand leads to representation (\ref{fourier_caplet_priceA}).

\end{proof}

\subsubsection{Model (b)}

For any $k \in \{1,\dots,n-1\}$, we have
\begin{align*}
f_K^{k,l}(x) \coloneqq \Bigg( & \tilde{D}^\mathsf{d}(T,T,T_k) \exp\Big( \sum_{j \in \mathscr{J}_T^{T_n}} \lambda_1^\mathsf{d}(T_{j-1}) x \Big) \\
& + \tilde{D}^l(T,T,T_k) \exp \Big( [ \sum_{j \in \mathscr{J}_T^{T_n}} \lambda_1^\mathsf{d}(T_{j-1}) + \bar{\gamma}_1^l(T)] x \Big) \\
& - \tilde{K}^l  A^\mathsf{d}(T,T_k) \exp \Big( \sum_{j \in \mathscr{J}_{T_k}^{T_n}} \lambda_1^\mathsf{d}(T_{j-1}) x \Big) \Bigg)^+
\end{align*}
and
\begin{align*}
f_K^{n,l}(x) = & \Bigg( D^\mathsf{d}(T,T,T_n) \exp \Big( \sum_{j \in \mathscr{J}_T^{T_n}} \lambda_1^\mathsf{d}(T_{j-1}) x \Big) \nonumber \\
& + D^l(T,T,T_n) \exp \Big( \Big[\sum_{j \in \mathscr{J}_T^{T_n}} \lambda_1^\mathsf{d}(T_{j-1}) + \bar{\gamma}_1^l(T)\Big] x \Big) - \tilde{K}^l \Bigg)^+,
\end{align*}
where we set
\begin{align*}
& \tilde{D}^\mathsf{d}(T,T,T_k) \coloneqq D^\mathsf{d}(T,T,T_k) A^\mathsf{d}(T,T_k) C^\mathsf{d}(T,T_k) \\
& \tilde{D}^l(T,T,T_k) \coloneqq D^l(T,T,T_k) A^\mathsf{d}(T,T_k) C^l(T,T_k)
\end{align*}
with
\begin{align*}
C^\mathsf{d}(T,T_k) \coloneqq \exp\Bigg( \int_0^T \sum_{j \in \mathscr{J}_T^{T_k}} \langle \lambda^\mathsf{d}(s,T_{j-1}), w(s,T_k,T^*) \rangle ds \Bigg)
\end{align*}
and
\begin{align*}
C^l(T,T_k) \coloneqq \exp\Bigg( \int_0^T  \langle \sum_{j \in \mathscr{J}_T^{T_k}} \lambda^\mathsf{d}(s,T_{j-1}) + \bar{\gamma}^l(s,T), w(s,T_k,T^*) \rangle ds \Bigg).
\end{align*}

\begin{proposition} \label{prop2}
The time-0 price of the caplet is given by
\begin{align} \label{fourier_caplet_priceB}
\mathsf{Cpl}(0,T,\delta^l,K) = \frac{B_0^\mathsf{d}(T_k)}{\pi} \int_0^\infty \mathsf{Re}\left(\varphi_{X_T}(u-iR) \hat{f}_K^{k,l}(iR-u)\right) du
\end{align}
where $\hat{f}_K^{k,l}$ denotes the extended Fourier transform of $f_K^{k,l}$ that admits the representation
\begin{align*}
\hat{f}_K^{k,l}(z) = e^{i z x_k} \Big[ & - \tilde{D}^\mathsf{d}(T,T,T_k) \frac{e^{\sum_{j \in \mathscr{J}_T^{T_n}} \lambda_1^\mathsf{d}(T_{j-1}) x_k}}{\sum_{j \in \mathscr{J}_T^{T_n}} \lambda_1^\mathsf{d}(T_{j-1}) + iz} \\
& -  \tilde{D}^l(T,T,T_k) \frac{e^{[\sum_{j \in \mathscr{J}_T^{T_n}} \lambda_1^\mathsf{d}(T_{j-1}) + \bar{\gamma}_1^l(T)] x_k}}{\sum_{j \in \mathscr{J}_T^{T_n}} \lambda_1^\mathsf{d}(T_{j-1}) + \bar{\gamma}_1^l(T) + iz} \\
& + \tilde{K}^l A^\mathsf{d}(T,T_k) \frac{e^{\sum_{j \in \mathscr{J}_{T_k}^{T_n}} \lambda_1^\mathsf{d}(T_{j-1}) x_k}}{\sum_{j \in \mathscr{J}_{T_k}^{T_n}} \lambda_1^\mathsf{d}(T_{j-1}) + iz} \Big]
\end{align*}
for any $k \in \{1,\dots,n-1\}$ and
\begin{align*}
\hat{f}_K^{n,l}(z) = e^{i z x_n} \Big[ & - D^\mathsf{d}(T,T,T_n) \frac{e^{\sum_{j \in \mathscr{J}_T^{T_n}} \lambda_1^\mathsf{d}(T_{j-1}) x_n}}{\sum_{j \in \mathscr{J}_T^{T_n}} \lambda_1^\mathsf{d}(T_{j-1}) + iz} \\
& - D^l(T,T,T_n) \frac{e^{[\sum_{j \in \mathscr{J}_T^{T_n}} \lambda_1^\mathsf{d}(T_{j-1}) + \bar{\gamma}_1^l(T)] x_n}}{\sum_{j \in \mathscr{J}_T^{T_n}} \lambda_1^\mathsf{d}(T_{j-1}) + \bar{\gamma}_1^l(T) + iz} + \frac{\tilde{K}^l}{iz} \Big],
\end{align*}
where $z \in \C$ with 
\begin{align*}
\sum_{j \in \mathscr{J}_T^{T_n}} \lambda_1^\mathsf{d}(T_{j-1}) + \bar{\gamma}_1^l(T) < \mathsf{Im}(z) 
\end{align*}
for all $T \in \mathscr{T}$, $x_k$ is the unique root of the function
\begin{align*}
h^{k,l}(x) \coloneqq & D^\mathsf{d}(T,T,T_k) C^\mathsf{d}(T,T_k) \exp\Big( \sum_{j \in \mathscr{J}_T^{T_k}} \lambda_1^\mathsf{d}(T_{j-1}) x \Big) \\
& + D^l(T,T,T_k) C^l(T,T_k) \exp \Big( [ \sum_{j \in \mathscr{J}_T^{T_k}} \lambda_1^\mathsf{d}(T_{j-1}) + \bar{\gamma}_1^l(T)] x \Big) \\
& - \tilde{K}^l,
\end{align*}
$x_n$ is the unique root of the function
\begin{align*}
h^{n,l}(x) \coloneqq & D^\mathsf{d}(T,T,T_n) \exp \Big( \sum_{j \in \mathscr{J}_T^{T_n}} \lambda_1^\mathsf{d}(T_{j-1}) x \Big) \nonumber \\
& +  D^l(T,T,T_n) \exp \Big( \Big[\sum_{j \in \mathscr{J}_T^{T_n}} \lambda_1^\mathsf{d}(T_{j-1}) + \bar{\gamma}_1^l(T)\Big] x \Big) - \tilde{K}^l 
\end{align*}
and $R=1+\frac{M-M^{'}}{M^{'}}$.
\end{proposition}
\begin{proof}
The proof is analogous to the proof of Proposition \ref{prop1}. In this case we define
\begin{align*}
\Lambda^\mathsf{d}(T,T_n) \coloneqq \sum_{j \in \mathscr{J}_T^{T_n}} \lambda_1^\mathsf{d}(T_{j-1}),
\end{align*}
and recall that
\begin{align*} 
\bar{\Lambda}^l(T,T_n) < R
\end{align*}
for all $T \in \mathscr{T}$. Then, for $k \in \{1,\dots,n-1\}$ the functions $g_K^{k,l}$ are bounded by
\begin{align*}
g_K^{k,l}(x) \leq \tilde{D}^\mathsf{d}(T,T,T_k) e^{(\Lambda^\mathsf{d}(T,T_n)-R) x_k} + \tilde{D}^l(T,T,T_k) e^{( \bar{\Lambda}^l(T,T_n) - R) x_k}
\end{align*}
and thus integrable. For $k=n$ the factors $\tilde{D}^\mathsf{d}$ and $\tilde{D}^l$ have to be replaced by $D^\mathsf{d}$ and $D^l$. The representation (\ref{fourier_caplet_priceB}) follows once again by the symmetry property of the integrand.
\end{proof}

\section{Model Calibration} 
 
\subsection{Data and Approach} 

We calibrate the model variants (a) and (b) to European market data observed in the post crisis period which is provided by Bloomberg. Market rates of deposits, forward rate agreements and swaps (OIS and Euribor indexed) based on different tenors as well as cap quotes indexed on Euribor for a number of maturities and strikes will be used. The cap quotes are given in form of the model dependent implied volatilities (in bps). More specifically we will consider cap volatility quotes indexed on six-month Euribor on September 15, 2016. Consequently, two term structures will be taken into account, namely the basic and the six-month curves. Since we observe a period with negative interest rates the standard log-normal market model can no longer be used in this case. Following market practice we use a multiple curve form of the Bachelier model in this situation to derive the market prices of caps from the volatility quotes. Finally the market prices of the caplets are derived from the cap prices. We also highlight that the considered cap contracts contain negative strike rates.

Based on multiple curve bootstrapping (see \citet{AmetranoBianchetti13} and \cite{GerhartLuetkebohmert18}) we construct tenor-dependent FRA curves. The basic curve is constructed by taking the quoted OIS rates. For each of the tenors we use the corresponding market quotes of deposit rates for the short-term maturities, rates from forward rate agreements for the mid-term part and swap rates for mid- and long-term maturities. Exact cubic spline interpolations are used during the bootstrap procedure. This approach guarantees also enough smoothness of the curves. 

We presented the bootstrapped basic, three-month and six-month FRA curves at the calibration date September 15, 2016 in Figure \ref{FRA_curves} at the bottom on the right-hand side. The values of the basic and the FRA curve for the six-month tenor as given there are used as input data in the calibration procedure. As formulas (\ref{fourier_caplet_priceA}) and (\ref{fourier_caplet_priceB}) show, the values of the basic curve along the tenor structure are needed. We present these values in Figure \ref{Basic_curve}. Note that the discount curve is increasing at the beginning because of the negativity of the rates.

Hereafter we describe the calibration procedure. Let $\Theta$ be the set of admissible model parameters, $\mathcal{T}$ the maturities and $\mathcal{K}$ the strike rates of the considered caps. We minimise the sum of the squared relative errors between model and mid market caplet prices
\begin{align*}
\sum_{T \in \mathcal{T},K \in \mathcal{K}} \left( \frac{\textsf{caplet model price}(\vartheta,T,K) - \textsf{caplet market price}(T,K)}{\textsf{caplet market price}(T,K)} \right)^2
\end{align*}
with respect to $\vartheta \in \Theta$. This optimisation is done by using a randomised Powell algorithm (see \cite{Powell78}). We use then the calibrated model parameters $\hat{\vartheta}$ to determine the model implied volatilities of caps. The differences between the implied volatilities of the model prices and the quoted volatilities specify the accuracy of the calibration. This procedure is done for both model variants (a) and (b).

\begin{figure}[htp]
\begin{center}
\hspace{-1.5cm}\includegraphics[width=6cm]{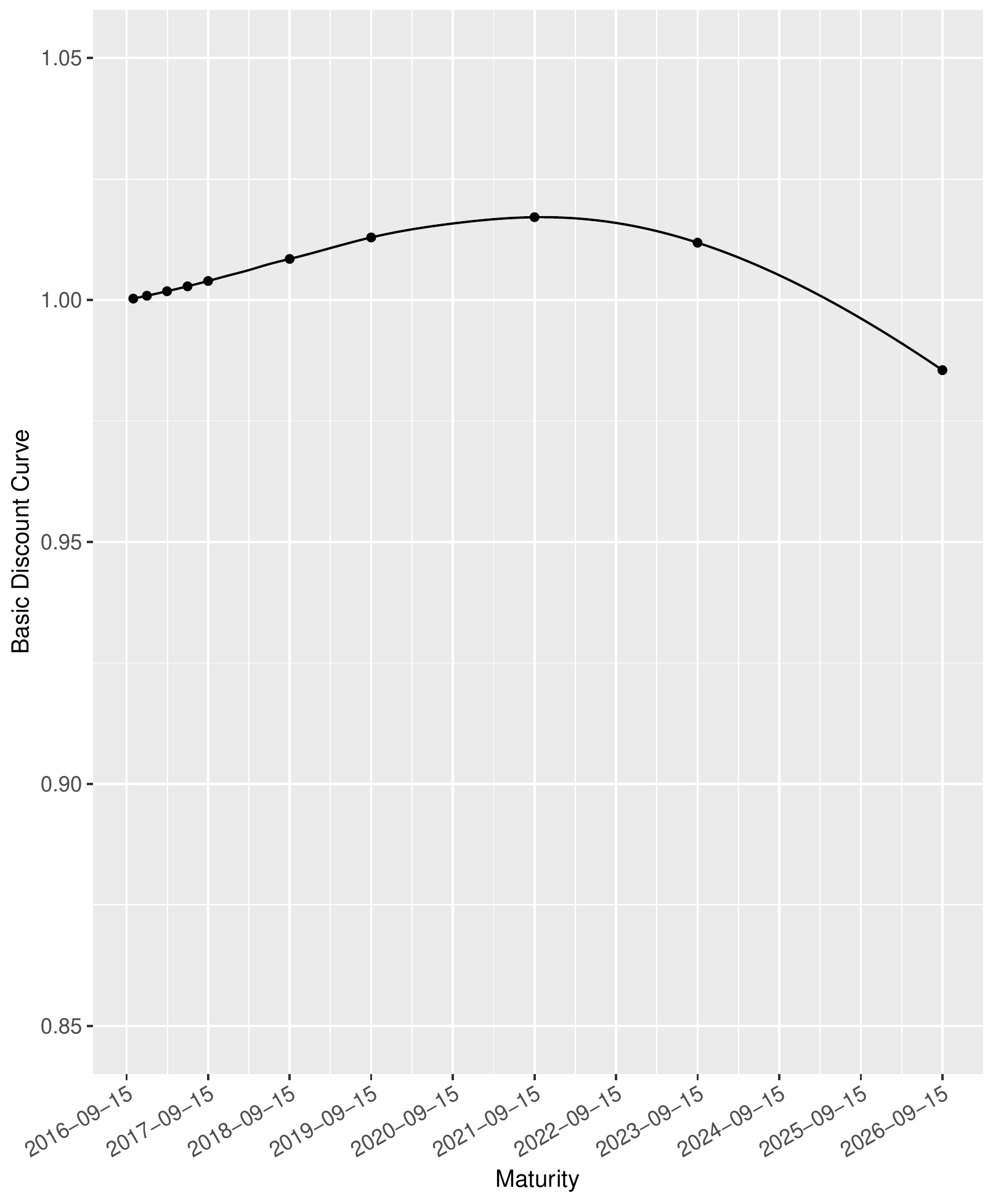}
\end{center}
\caption{Bootstrapped basic curve on September 15, 2016.}
\label{Basic_curve}
\end{figure}

\subsection{Model Specification and Calibration}

First let us specify the driving process $L^{T^*}$ under $P_{T^*}^\mathsf{d}$. We will use a normal inverse Gaussian (NIG) Lévy process with parameters $\alpha$, $\beta$, $\delta$ and $\mu$ (see for example \cite{Eberlein09}) which have to satisfy  $0 \leq\abs{\beta} < \alpha $, $\delta > 0$ and $\mu \in \R$. The last parameter $\mu$ does not enter into the valuation formulas. We  choose it such that the expectation of $L_1^{T^*}$ is equal to zero. The distributions of a NIG process are completely determined by its cumulant function
\begin{align*}
\theta(z) = \mu z + \delta \left(\sqrt{\alpha^2 - \beta^2} - \sqrt{\alpha^2 - (\beta + z)^2}\right)
\end{align*} 
where $\mathsf{Re}(z) \in (-\alpha-\beta,\alpha - \beta)$. We emphasize that only parameters which lead to a Lévy measure $F^{T^*}$ that satisfies Assumption ($\mathbb{EM}$) are admissible.

According to Assumption ($\mathbb{VOL}$) the volatility structures are defined if we specify $\lambda$, $\lambda_1^\mathsf{d}$, $\gamma_1^l$ and $\bar{\gamma}_1^l$. We choose
\begin{align*}
\lambda(t)=\exp(a t), \quad \lambda_1^\mathsf{d}(T)=\sqrt{\abs{a_\mathsf{d}} T}, \quad \gamma_1^l(T)=\sqrt{\abs{a_l} T}, \quad \bar{\gamma}_1^l(T)=\sqrt{\abs{\bar{a}_l} T}.
\end{align*}
Consequently four real-valued parameters $a$, $a_\mathsf{d}$, $a_l$ and $\bar{a}_l$ describe the volatilities. We emphasize that the volatility functions have to satisfy boundedness restrictions according to Assumption ($\mathbb{VOL}$). Thus calibration requires a nonlinear optimization under several constraints.

In both graphs of Figure \ref{Calibration} the grid represents the market volatility surface on September 15, 2016. The points in the graphs indicate the implied volatilities of the calibrated model. In particular the graphs show also that both models are able to cope with negative strike rates as well as negative interest rates which prevailed in September 2016. The parameters corresponding to the calibrated models are given in Table \ref{calibrated_parameters_2016}.

\begin{figure}
\hfill
\subfigure[]{\includegraphics[width=5cm]{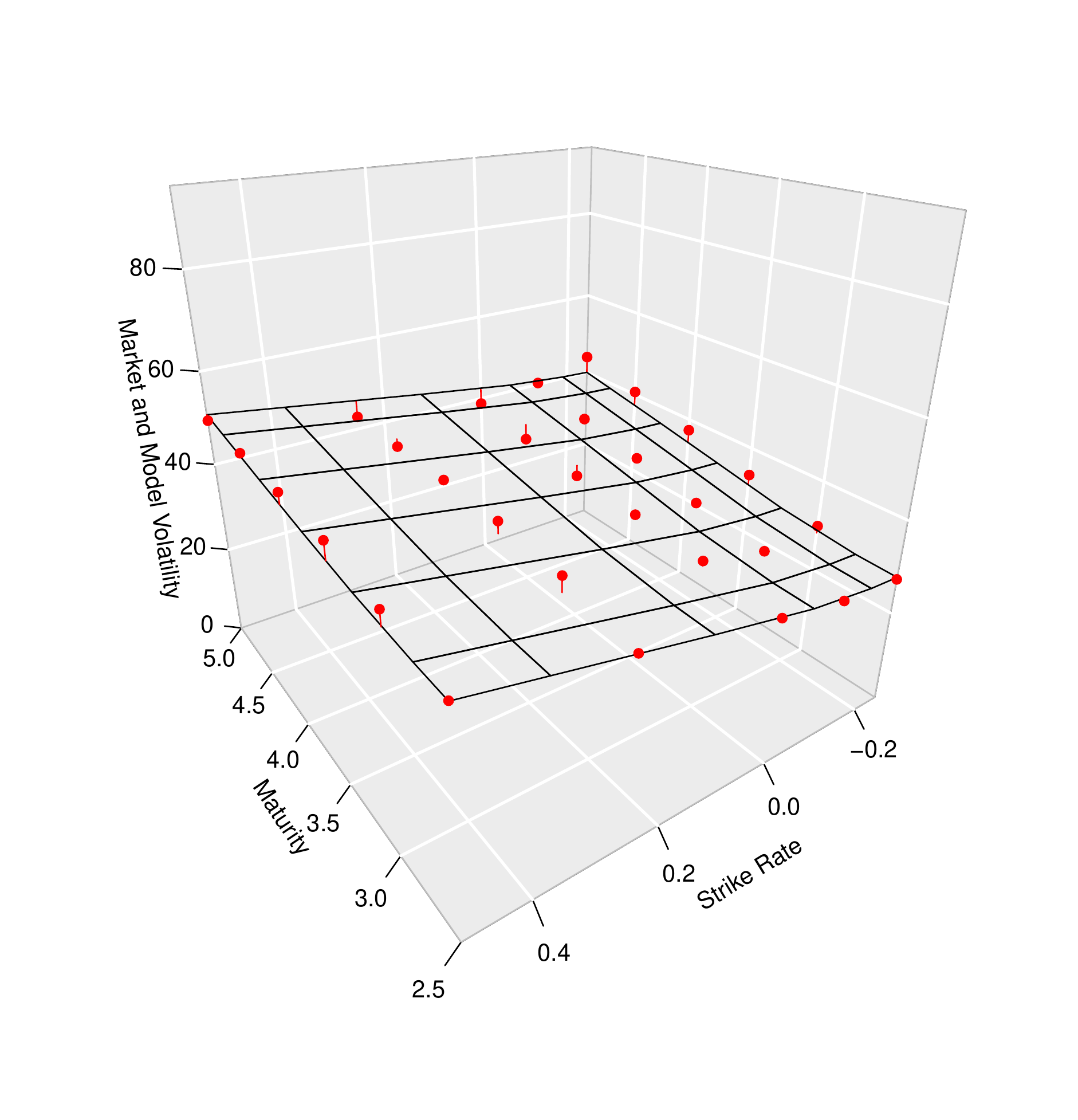}}
\hfill
\subfigure[]{\includegraphics[width=5cm]{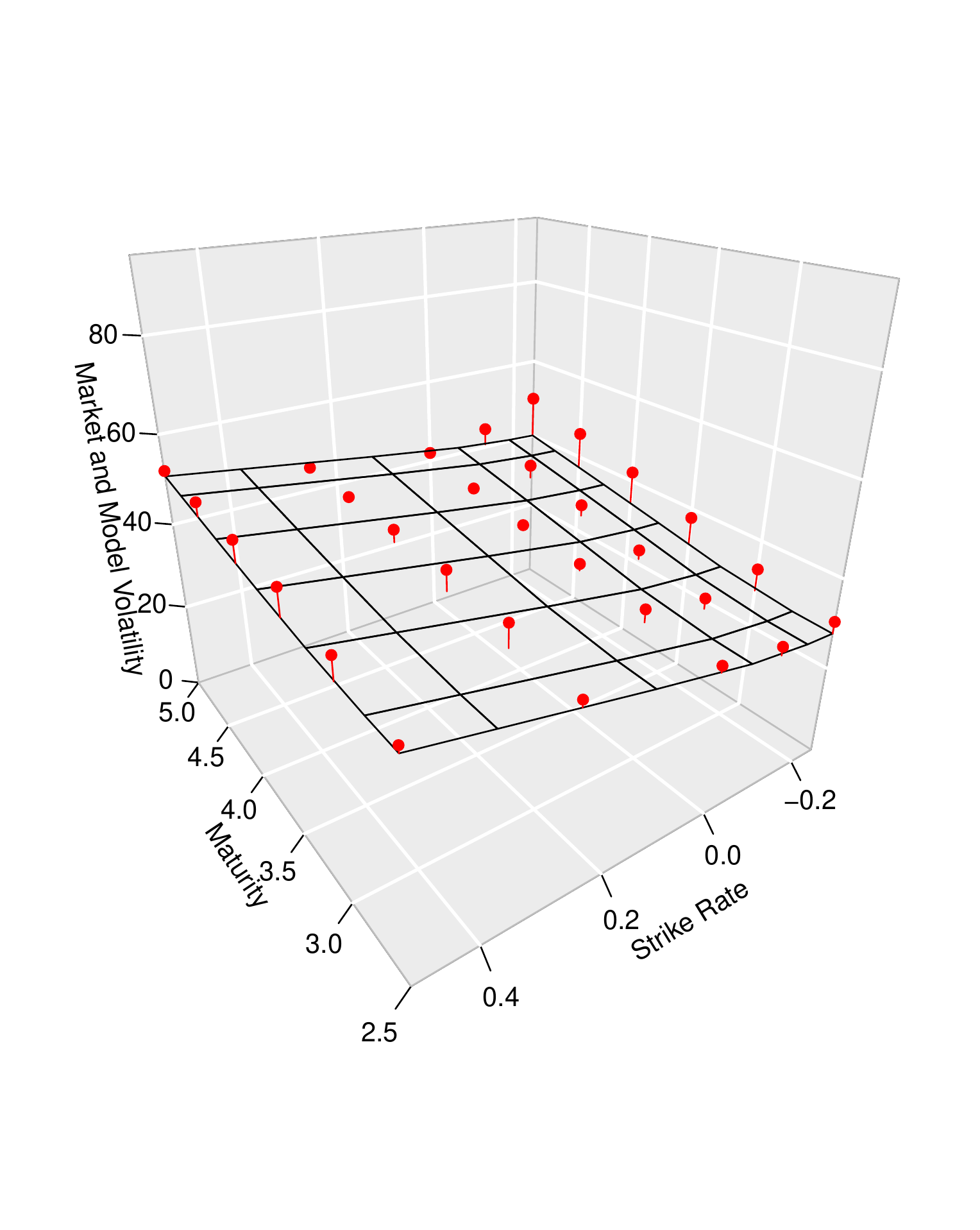}}
\hfill
 \caption{Calibration results of both model variants on September 15, 2016.}
\label{Calibration}
\end{figure}

\begin{table}[htbp]
\centering
\begin{tabular}{@{}lr@{\hspace*{2cm}}lr@{}}
\toprule[1pt]
\multicolumn{4}{c}{\textbf{Calibrated Parameters}}\\
\hline
\hline
\multicolumn{4}{c}{Model Variant (a)}\\
\hline
~\\[-1.3ex]
\multicolumn{2}{@{}l}{\textbf{NIG}}&\multicolumn{2}{@{}l}{\textbf{Volatility Structure}}\\
$\alpha$&53.66666& $a$& $-3.498342$\\
$\beta$&-47.62499&$a_{\mathsf{d}}$& $-0.009348$\\
$\delta$&0.105083&$a_l$& $0.000548$ \\
\multicolumn{4}{c}{Model Variant (b)}\\
\hline
~\\[-1.3ex]
\multicolumn{2}{@{}l}{\textbf{NIG}}&\multicolumn{2}{@{}l}{\textbf{Volatility Structure}}\\
$\alpha$&  2.35391 & $a$& -6.003533 \\
$\beta$&  0.87951&$a_\mathsf{d}$& 0.002264 \\
$\delta$& 14.6241 &$\bar{a}_l$& 0.001549 \\
\bottomrule
\end{tabular}
 \caption{September 15, 2016.}
  \label{calibrated_parameters_2016}
\end{table}


\newpage

\bibliography{Literatur}

\begin{thebibliography}{24}
\providecommand{\natexlab}[1]{#1}
\providecommand{\url}[1]{\texttt{#1}}
\expandafter\ifx\csname urlstyle\endcsname\relax
  \providecommand{\doi}[1]{doi: #1}\else
  \providecommand{\doi}{doi: \begingroup \urlstyle{rm}\Url}\fi

\bibitem[Ametrano and Bianchetti(2013)]{AmetranoBianchetti13}
Ferdinando~M. Ametrano and Marco Bianchetti.
\newblock {Everything You Always Wanted to Know About Multiple Interest Rate
  Curve Bootstrapping But Were Afraid to Ask}.
\newblock \emph{Available at SSRN 2219548}, 2013.

\bibitem[Beinhofer et~al.(2011)Beinhofer, Eberlein, Janssen, and
  Polley]{BeinhoferEberlein2011}
Maximilian Beinhofer, Ernst Eberlein, Arend Janssen, and Manuel Polley.
\newblock {Correlations in L{\'e}vy Interest Rate Models}.
\newblock \emph{Quantitative Finance}, 11\penalty0 (9):\penalty0 1315--1327,
  2011.

\bibitem[Bianchetti and Morini(2013)]{BianchettiMorini13}
Marco Bianchetti and Massimo Morini.
\newblock {Interest Rate Modelling After the Financial Crisis}.
\newblock \emph{Risk Books}, 2013.

\bibitem[Brigo and Mercurio(2006)]{BrigoMercurio06}
Damiano Brigo and Fabio Mercurio.
\newblock \emph{{Interest Rate Models - Theory and Practice: With Smile,
  Inflation and Credit}}.
\newblock Springer Finance. Springer, second edition, 2006.

\bibitem[Cr{\'e}pey et~al.(2012)Cr{\'e}pey, Grbac, and
  Nguyen]{CrepeyGrbacNguyen12}
St{\'e}phane Cr{\'e}pey, Zorana Grbac, and Hai-Nam Nguyen.
\newblock {A Multiple-Curve HJM Model of Interbank Risk}.
\newblock \emph{Mathematics and Financial Economics}, 6\penalty0 (3):\penalty0
  155--190, 2012.

\bibitem[Cuchiero et~al.(2016)Cuchiero, Fontana, and
  Gnoatto]{CuchieroFontanaGnoatto16}
Christa Cuchiero, Claudio Fontana, and Alessandro Gnoatto.
\newblock {General HJM Framework for Multiple Yield Curve Modeling}.
\newblock \emph{Finance and Stochastics}, 20\penalty0 (2):\penalty0 267--320,
  2016.

\bibitem[Cuchiero et~al.(2017)Cuchiero, Fontana, and
  Gnoatto]{CuchieroFontanaGnoatto17}
Christa Cuchiero, Claudio Fontana, and Alessandro Gnoatto.
\newblock {Affine multiple yield curve models}.
\newblock \emph{Mathematical Finance (forthcoming)}, 2017.

\bibitem[Eberlein(2009)]{Eberlein09}
Ernst Eberlein.
\newblock {Jump--Type L{\'e}vy Processes}.
\newblock In \emph{Handbook of Financial Time Series}, pages 439--455.
  Springer, 2009.

\bibitem[Eberlein and Gerhart(2018)]{EberleinGerhart17}
Ernst Eberlein and Christoph Gerhart.
\newblock {A Multiple-Curve L{\'e}vy Forward Rate Model in a Two-Price
  Economy}.
\newblock \emph{Quantitative Finance}, 18\penalty0 (4):\penalty0 537--561,
  2018.

\bibitem[Eberlein and {\"O}zkan(2005)]{EberleinOezkan05}
Ernst Eberlein and Fehmi {\"O}zkan.
\newblock {The L{\'e}vy Libor Model}.
\newblock \emph{Finance and Stochastics}, 9\penalty0 (3):\penalty0 327--348,
  2005.

\bibitem[Eberlein and Raible(1999)]{EberleinRaible99}
Ernst Eberlein and Sebastian Raible.
\newblock {Term Structure Models Driven by General L{\'e}vy Processes}.
\newblock \emph{Mathematical Finance}, 9\penalty0 (1):\penalty0 31--53, 1999.

\bibitem[Eberlein et~al.(2005)Eberlein, Jacod, and
  Raible]{EberleinJacodRaible05}
Ernst Eberlein, Jean Jacod, and Sebastian Raible.
\newblock {L{\'e}vy Term Structure Models: No-Arbitrage and Completeness}.
\newblock \emph{Finance and Stochastics}, 9\penalty0 (1):\penalty0 67--88,
  2005.

\bibitem[Eberlein et~al.(2010)Eberlein, Glau, and
  Papapantoleon]{EberleinGlauPapapantoleon10}
Ernst Eberlein, Kathrin Glau, and Antonis Papapantoleon.
\newblock {Analysis of Fourier Transform Valuation Formulas and Applications}.
\newblock \emph{Applied Mathematical Finance}, 17\penalty0 (3):\penalty0
  211--240, 2010.

\bibitem[Eberlein et~al.(2016)Eberlein, Eddahbi, and Lalaoui
  Ben~Cherif]{EberleinEddahbiLalaoui16}
Ernst Eberlein, M'hamed Eddahbi, and Sidi~Mohamed Lalaoui Ben~Cherif.
\newblock {Option Pricing and Sensitivity Analysis in the Lévy Forward Process
  Model}.
\newblock In K.~Glau, Z.~Grbac, M.~Scherer, and R.~Zagst, editors,
  \emph{Innovations in Derivatives Markets}, pages 285--313. Springer, 2016.

\bibitem[Gerhart and Lütkebohmert(2018)]{GerhartLuetkebohmert18}
Christoph Gerhart and Eva Lütkebohmert.
\newblock {Empirical Analysis and Forecasting of Multiple Yield Curves}.
\newblock \emph{{Preprint University of Freiburg}}, 2018.

\bibitem[Grbac and Runggaldier(2015)]{GrbacRunggaldier15}
Zorana Grbac and Wolfgang~J. Runggaldier.
\newblock \emph{{Interest Rate Modeling: Post-Crisis Challenges and
  Approaches}}.
\newblock SpringerBriefs in Quantitative Finance, 2015.

\bibitem[Grbac et~al.(2015)Grbac, Papapantoleon, Schoenmakers, and
  Skovmand]{Grbac15affine}
Zorana Grbac, Antonis Papapantoleon, John Schoenmakers, and David Skovmand.
\newblock {Affine LIBOR Models with Multiple Curves: Theory, Examples and
  Calibration}.
\newblock \emph{SIAM Journal on Financial Mathematics}, 6\penalty0
  (1):\penalty0 984--1025, 2015.

\bibitem[Henrard(2010)]{Henrard10}
Marc Henrard.
\newblock {The Irony in Derivatives Discounting Part II: The Crisis}.
\newblock \emph{Wilmott Journal}, 2\penalty0 (6):\penalty0 301--316, 2010.

\bibitem[Henrard(2014)]{Henrard14}
Marc Henrard.
\newblock {Interest Rate Modelling in the Multi-curve Framework}.
\newblock \emph{Palgrave Macmillan}, 2014.

\bibitem[Jacod and Shiryaev(2003)]{JacodShiryaev03}
Jean Jacod and Albert Shiryaev.
\newblock \emph{{Limit Theorems for Stochastic Processes}}.
\newblock Grundlehren der mathematischen Wissenschaften. Springer-Verlag,
  second edition, 2003.

\bibitem[Kallsen and Shiryaev(2002)]{KallsenShiryaev02}
Jan Kallsen and Albert~N. Shiryaev.
\newblock {The Cumulant Process and Esscher's Change of Measure}.
\newblock \emph{Finance and Stochastics}, 6\penalty0 (4):\penalty0 397--428,
  2002.

\bibitem[Mercurio(2009)]{Mercurio09}
Fabio Mercurio.
\newblock {Interest Rates and the Credit Crunch: New Formulas and Market
  Models}.
\newblock \emph{Bloomberg portfolio research paper}, 2009.

\bibitem[Powell(1978)]{Powell78}
Michael~J.D. Powell.
\newblock {A Fast Algorithm for Nonlinearly Constrained Optimization
  Calculations}.
\newblock In \emph{{Numerical Analysis}}, pages 144--157. Springer, 1978.

\bibitem[Sato(1999)]{Sato99}
Ken-iti Sato.
\newblock \emph{{L{\'e}vy Processes and Infinitely Divisible Distributions}}.
\newblock {Cambridge University Press}, 1999.

\end{thebibliography}

\bibliographystyle{plainnat}

\end{document}